\documentclass{IEEEtran}
\usepackage{amsmath}
\usepackage{amsfonts}
\usepackage{amssymb}
\usepackage{float}
\usepackage{xcolor}
\usepackage{enumerate}
\restylefloat{table}
\usepackage{tikz} 
\usetikzlibrary{arrows, positioning, automata}
\usetikzlibrary{shapes.geometric, arrows}
\newcommand*{\Comb}[2]{{}^{#1}C_{#2}}

		\newtheorem{converse}{Converse}[section]

\def\SC{{\rm SC}}
\def\OPT{{\rm OPT}}

\def\LP{{\rm LP}}

\def\DP{{\rm DP}}

\def\OBJ{{\rm OBJ}}
\def\DPSW{{\rm DPSW}}
\def\SW{{\rm SW}}
\def\SID{{\rm SID}}

\def\DPSW{{\rm DPSW}}

\def\DPSI{{\rm DPSI}}
\def\DPJE{{\rm DPJE}}

\def\LPSI{{\rm LPSI}}

\def\SW{{\rm SW}}
\def \gammab{\bar{\gamma}}

\def \lambdab{\bar{\lambda}}
\def\lambdat {\tilde{\lambda}}
\def \gammat{\tilde{\gamma}}
\def \lambdah{\widehat{\lambda}}
\def \phih{\widehat{\phi}}
\def \gammah{\widehat{\gamma}}
\newcommand\independent{\protect\mathpalette{\protect\independenT}{\perp}}
\def\independenT#1#2{\mathrel{\rlap{$#1#2$}\mkern2mu{#1#2}}}
\def\csf{\mathsf{c}}
\def\ssf{\mathsf{s}}

\usepackage{amsmath, amssymb, bbm, xspace}
\usepackage{epsfig}
\usepackage{longtable}
\usepackage{color}
\usepackage{mathrsfs}
\usepackage{comment}

\usepackage{courier}

\newtheorem{theorem}{Theorem}[section]

\newtheorem{lemma}{Lemma}[section]

\newtheorem{proposition}[theorem]{Proposition}

\newtheorem{corollary}[theorem]{Corollary}

\def\bkE{{\rm I\kern-.17em E}}
\def\bk1{{\rm 1\kern-.17em l}}
\def\bkD{{\rm I\kern-.17em D}}
\def\bkR{{\rm I\kern-.17em R}}
\def\bkP{{\rm I\kern-.17em P}}

\def\bkZ{{\bf{Z}}}

\def\bkE{{\rm I\kern-.17em E}}
\def\bk1{{\rm 1\kern-.17em l}}
\def\bkD{{\rm I\kern-.17em D}}
\def\bkR{{\rm I\kern-.17em R}}
\def\bkP{{\rm I\kern-.17em P}}

\makeatletter
\newcommand{\pushright}[1]{\ifmeasuring@#1\else\omit\hfill$\displaystyle#1$\fi\ignorespaces}
\newcommand{\pushleft}[1]{\ifmeasuring@#1\else\omit$\displaystyle#1$\hfill\fi\ignorespaces}
\makeatother

\def\bkZ{{\bf{Z}}}
\def\b12{(\beta_1,\beta_2)}

\newenvironment{proofarg}[1][]{\noindent\hspace{2em}{\itshape Proof #1: }}{\hspace*{\fill}~\qed\par\endtrivlist\unskip}

\newcounter{example}
\renewcommand{\theexample}{\thesection.\arabic{example}}

\newcounter{remark}
\renewcommand{\theremark}{\thesection.\arabic{remark}}
\newenvironment{remarkc}[1][]{\refstepcounter{remark}
\noindent{\itshape Remark~\theremark. #1} \rmfamily}{\hspace*{\fill}~$\square$\vspace{0pt}}

\def\Bscr{\mathscr{B}}
\def\Xscr{\mathcal{X}}
\def\Yscr{\mathcal{Y}}

\def\Ebb{\mathbb{E}}
\newlength{\noteWidth}
\long\def\notes#1{\ifinner
{\tiny #1}
\else
\marginpar{\parbox[t]{\noteWidth}{\raggedright\tiny #1}}
\fi\typeout{#1}}

 \def\notes#1{\typeout{read notes: #1}} 

\newcommand{\I}[1]{\mathbb{I}_{\{#1\}}}

\newcommand{\ie}{i.e.\@\xspace} 

\newcommand{\Real}{\ensuremath{\mathbb{R}}}

\newcommand{\minimize}[1]{\displaystyle\minim_{#1}}
\newcommand{\minim}{\mathop{\hbox{\rm min}}}
\newcommand{\maximize}[1]{\displaystyle\maxim_{#1}}
\newcommand{\maxim}{\mathop{\hbox{\rm max}}}

\def\OPT{{\rm OPT}}
\def\FEA{{\rm FEA}}

\def\Ebb{\mathbb{E}}

\def\Pbb{{\mathbb{P}}}
\def\Nbb{{\mathbb{N}}}
\def\Qbb{{\mathbb{Q}}}
\def\Ibb{{\mathbb{I}}}

\def\dbf{{\bf d}}

\def\exp{\mathop{\hbox{\rm exp}}}  
\def\In{\mathop{\hbox{\it In}\,}}

\def\norm#1{\|#1\|}

\def\spose#1{\hbox to 0pt{#1\hss}}
\def\sub#1{^{\null}_{#1}}
\def\text #1{\hbox{\quad#1\quad}}

\def\zhat{{\hat x}}

\def\nthinsp{\mskip -2   mu}

\def\Q{_{\scriptscriptstyle Q}}

\def\superstar{^{\raise 0.5pt\hbox{$\nthinsp *$}}}
\def\SUPERSTAR{^{\raise 0.5pt\hbox{$*$}}}

\def\lamstarT {\lambda^{\raise 0.5pt\hbox{$\nthinsp *$}T}}

\def\shat{\widehat s}

\def\Ascr{{\cal A}}
\def\Bscr{{\cal B}}

\def\Pscr{{\cal P}}

\def\Sscr{{\cal S}}
\def\Sscrhat{\widehat{\mathcal{S}}}

\def\Zscr{{\cal Z}}
\def\Xscr{{\cal X}}
\def\Yscr{{\cal Y}}

\def\shat{\widehat s}

\def\Shat{\widehat S}

\def\zbar{\skew{2.8}\bar z}
\def\zhat{\skew{2.8}\widehat z}
\def\ztilde{\skew{2.8}\widetilde z}

\def\sbf{{\bf s}}

\def\non{\nonumber}

\let\forallnew\forall
\renewcommand{\forall}{\forallnew\ }
\let\forall\forallnew

		\def\bkE{{\rm I\kern-.17em E}}
		\def\bk1{{\rm 1\kern-.17em l}}
		\def\bkD{{\rm I\kern-.17em D}}
		\def\bkR{{\rm I\kern-.17em R}}
		\def\bkP{{\rm I\kern-.17em P}}
		\def\bkY{{\bf \kern-.17em Y}}
		\def\bkZ{{\bf \kern-.17em Z}}
		\def\bkC{{\bf  \kern-.17em C}}

{\begin{list}{}%
         {\setlength{\leftmargin}{#1}}%
         \item[]%
}
{\end{list}}

		\def\bsp{\begin{split}}
		\def\beq{\begin{eqnarray}}
		\def\bal{\begin{align*}}
		\def\bc{\begin{center}}
		\def\be{\begin{enumerate}}
		\def\bi{\begin{itemize}}
		\def\bs{\begin{small}}
		\def\bS{\begin{slide}}
		\def\ec{\end{center}}
		\def\ee{\end{enumerate}}
		\def\ei{\end{itemize}}
		\def\es{\end{small}}
		\def\eS{\end{slide}}
		\def\eeq{\end{eqnarray}}
		\def\eal{\end{align*}}
		\def\esp{\end{split}}
		\def\qed{ \vrule height7.5pt width7.5pt depth0pt}  

	\def\maxproblemlarge#1#2#3#4{\fbox
		 {\begin{tabular*}{1.0\textwidth}
			{@{}l@{\extracolsep{15pt}}l@{\extracolsep{6pt}}l@{\extracolsep{\fill}}c@{}}
				#1 & $\maximize{#2}$ & $#3$ & $ $ \\[5pt]
					$\sub\ $ &     & $#4$ & $ $
			\end{tabular*}}
			}
	\def\maxproblemsmall#1#2#3#4{\fbox
		 {\begin{tabular*}{0.47\textwidth}
			{@{}l@{\extracolsep{\fill}}l@{\extracolsep{6pt}}l@{\extracolsep{\fill}}c@{}}
				#1 & $\maximize{#2}$ & $#3$ & $ $ \\[5pt]
					 & $\subject\ $    & $#4$ & $ $
			\end{tabular*}}
			}

\def\problemsmallLP#1#2#3#4{\fbox
		 {\begin{tabular*}{1.0\textwidth}
			{@{}l@{\extracolsep{\fill}}l@{\extracolsep{6pt}}l@{\extracolsep{\fill}}c@{}}
				#1 & $\minimize{#2}$ & $#3$ & $ $ \\[5pt]
					  $\sub \ $ &    & $#4$ & $ $
			\end{tabular*}}
			}

	\def\cp2problem#1#2#3#4{\fbox
		 {\begin{tabular*}{0.9\textwidth}
			{@{}l@{\extracolsep{\fill}}l@{\extracolsep{6pt}}l@{\extracolsep{\fill}}c@{}}
				#1 & & $#4 $ 
			\end{tabular*}}}

		\def\bkE{{\rm I\kern-.17em E}}
		\def\bk1{{\rm 1\kern-.17em l}}
		\def\bkD{{\rm I\kern-.17em D}}
		\def\bkR{{\rm I\kern-.17em R}}
		\def\bkP{{\rm I\kern-.17em P}}
		
		\def\bkZ{{\bf{Z}}}

\newcommand {\beeq}[1]{\begin{equation}\label{#1}}
\newcommand {\eeeq}{\end{equation}}
\newcommand {\bea}{\begin{eqnarray}}
\newcommand {\eea}{\end{eqnarray}}

\def\texitem#1{\par\smallskip\noindent\hangindent 25pt
               \hbox to 25pt {\hss #1 ~}\ignorespaces}

\def\bsp{\begin{split}}
		\def\beq{\begin{eqnarray}}
		\def\bal{\begin{align*}}
		\def\bc{\begin{center}}
		\def\be{\begin{enumerate}}
		\def\bi{\begin{itemize}}
		\def\bs{\begin{small}}
		\def\bS{\begin{slide}}
		\def\ec{\end{center}}
		\def\ee{\end{enumerate}}
		\def\ei{\end{itemize}}
		\def\es{\end{small}}
		\def\eS{\end{slide}}
		\def\eeq{\end{eqnarray}}
		\def\eal{\end{align*}}
		\def\esp{\end{split}}
		\def\qed{ \vrule height7.5pt width7.5pt depth0pt}  

\usepackage{amsmath, amssymb, xspace}
\usepackage{epsfig}
\usepackage{longtable}
\usepackage{color}
\usepackage{mathrsfs}
\usepackage{subfig}
\newenvironment{proof}[1][]{{\noindent \bf Proof #1: }}{\hfill \qed \vspace{3pt}\\ }

\def\sub{\hbox{\rm s.t.}}

			\def\problemsmalla#1#2#3#4{\fbox
		 {\begin{tabular*}{0.475\textwidth}
			{@{}l@{\extracolsep{\fill}}l@{\extracolsep{-4pt}}l@{\extracolsep{\fill}}c@{}}
				#1 &  & $\minimize{#2}$  $#3$ & $ $ \\[4pt]
					  $\sub \ $  &   & $#4$ &  $ $
			\end{tabular*}}
			}
			
			\def\singlecolproblemsmalla#1#2#3#4{\fbox
		 {\begin{tabular*}{0.480\textwidth}
			{@{}l@{\extracolsep{\fill}}l@{\extracolsep{-4pt}}l@{\extracolsep{\fill}}c@{}}
				#1 &  & $\minimize{#2}$  $#3$ & $ $ \\[4pt]
					  $\sub \ $  &   & $#4$ &  $ $
			\end{tabular*}}
			}

	\def\maxproblemsmall#1#2#3#4{\fbox
		 {\begin{tabular*}{0.47\textwidth}
			{@{}l@{\extracolsep{\fill}}l@{\extracolsep{-4pt}}l@{\extracolsep{\fill}}c@{}}
				#1 &  & $\maximize {#2}$ $#3$ & $ $ \\[4pt]
					 $\sub \ $  &   & $#4$ & $ $
			\end{tabular*}}
			}

			\renewcommand{\I}[1]{\mathbb{I}\{#1\}}
\ifCLASSINFOpdf
  \else
 
\fi

\author{\IEEEauthorblockN{Sharu Theresa Jose \quad Ankur A. Kulkarni}\thanks{Sharu and Ankur are with the Systems and Control Engineering group at the Indian Institute of Technology Bombay, Mumbai, 400076, India. 
email: sharutheresa@iitb.ac.in,\quad kulkarni.ankur@iitb.ac.in. This work was presented in part at the IEEE Information Theory Workshop, in Kaohsiung, Taiwan, 2017~\cite{jose2017linearITW}.}}
\title{Improved Finite Blocklength Converses for Slepian-Wolf Coding via Linear Programming}
\begin{document}
\maketitle
\begin{abstract}
A new finite blocklength converse for the Slepian-Wolf coding problem is presented which significantly improves on the best known converse for this problem, due to Miyake and Kanaya~\cite{miyake1995coding}. To obtain this converse, an extension of the linear programming (LP) based framework for finite blocklength point-to-point coding problems from \cite{jose2016linear}  is employed. 
However, a direct application of this framework demands a complicated analysis for the Slepian-Wolf problem. 
An analytically simpler approach is presented wherein LP-based finite blocklength converses for this problem are synthesized from  point-to-point lossless source coding problems  with perfect side-information at the decoder. 
New finite blocklength metaconverses for these point-to-point problems are derived by employing the LP-based framework, and the new converse for Slepian-Wolf coding is obtained by an appropriate combination of these converses.
\end{abstract}
\section{Introduction}
The intractability of evaluating the nonasymptotic or finite blocklength fundamental limit of communication has put the onus on discovering finite blocklength achievability and converses that sandwich tightly the nonasymptotic fundamental limit. Accordingly, recent years have witnessed a surge of tight finite blocklength achievability and converses (\cite{polyanskiy2010channel}, \cite{kostina2012fixed}, \cite{kostina2013lossy}, \cite{jose2016linear}), particularly for coding problems in the point-to-point setting. 

Eventhough many sharp and asymptotically tight finite blocklength converses have been obtained in the point-to-point setting employing tools like hypothesis testing \cite{polyanskiy2010channel} and information spectrum \cite{han2003information}, deriving tight finite blocklength converses for multiterminal coding problems still remains particularly challenging. Part of this challenge could be attributed to the difficulty in extending the techniques in the point-to-point setting to the network setting.
In this paper, we consider the  classical multiterminal source coding problem -- the Slepian-Wolf coding problem and show that the extension of the linear programming (LP) based framework we introduced for the point-to-point setting in \cite{jose2016linear}, in fact results in new and improved finite blocklength converses for this problem. Moreover, it yields a framework via a hierarchy of relaxations in which classical converses can be recovered, and converses for the networked problem can be synthesized using a combination of point-to-point converses.

\begin{figure}[h]
\begin{center}
\includegraphics[scale=0.33,clip=true, trim = 0in 5.3in 0in 2.5in]{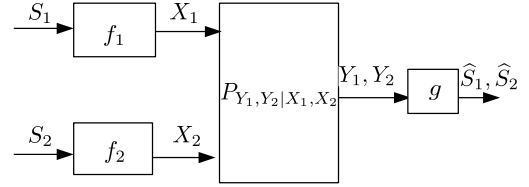} 
\caption{Two-User Joint Source-Channel Coding Problem}\label{fig:SW}
\vspace{-0.6cm}
\end{center}
\end{figure}

Consider the finite blocklength Slepian-Wolf distributed lossless source coding problem (in Figure~\ref{fig:SW}) posed as the following optimization problem,
$$\problemsmalla{SW}
	{f_1,f_2,g}
	{\displaystyle \mathbb{E}[\Ibb\{(S_1,S_2) \neq (\Shat_1,\Shat_2)\}]}
				 {\begin{array}{r@{\ }c@{\ }l}
				 				 X_1&=&f_1(S_1),\\
				 				 X_2&=&f_2(S_2),\\
								 (\Shat_1,\Shat_2)&=&g(Y_1,Y_2),
	\end{array}}$$
	where $S_1,S_2,X_1,X_2,Y_1,Y_2,\Shat_1,\Shat_2$ are discrete random variables taking values in fixed, finite spaces $\Sscr_1,\Sscr_2,\Xscr_1,\Xscr_2,\Yscr_1,\Yscr_2,\Sscrhat_1,\Sscrhat_2$, respectively. Notice that these spaces could be Cartesian products of smaller spaces, and hence could be sets of finite length strings. Here, 
$S_1$ and $S_2$ represent the two \textit{correlated} sources distributed according to a known joint probability distribution $P_{S_1,S_2}$. The source signals are \textit{seperately} encoded  by functions $f_1:\Sscr_1\rightarrow \Xscr_1$ and $f_2:\Sscr_2\rightarrow \Xscr_2$ to produce  signals $X_1=f_1(S_1)$ and $X_2=f_2(S_2)$, respectively. The encoded signals are sent through a deterministic channel with conditional distribution $P_{Y_1,Y_2|X_1,X_2}= \Ibb\{(Y_1,Y_2)=(X_1,X_2)\}$ to get the signal $(Y_1,Y_2)$, where $\Ibb\{\cdot\}$ represents the indicator function which equals unity when `$\cdot$' is true and is zero otherwise.
 $(Y_1,Y_2)$ is then \textit{jointly} decoded by $g:\Yscr_1 \times \Yscr_2 \rightarrow \Sscrhat_1\times \Sscrhat_2$ to obtain the output signal $(\Shat_1,\Shat_2)=g(Y_1,Y_2).$ 
  For the finite blocklength Slepian-Wolf coding problem, we note that spaces $\Sscr_1=\Sscrhat_1$, $\Sscr_2=\Sscrhat_2$, $\Xscr_1=\Yscr_1=\{1,\hdots,M_1\}$ and $\Xscr_2=\Yscr_2=\{1,\hdots, M_2\}$, $M_1,M_2 \in \Nbb$. 
  An error in transmission occurs when $(S_1,S_2)\neq (\Shat_1,\Shat_2)$. Hence, the objective of the finite blocklength Slepian-Wolf coding problem SW is to minimize the probability of error over all \textit{codes}, \ie over all encoder-decoder functions $(f_1,f_2,g)$.	
  
Our interest in this paper is in obtaining finite blocklength converses (or lower bounds) on the optimal value of SW and our approach is via the linear programming (LP) based framework introduced in \cite{jose2016linear}.
In \cite{jose2016linear}, we showed that this framework recovers and  improves on most of the well-known finite blocklength converses for point-to-point coding problems. In particular, the LP framework is shown to imply the metaconverse of Polyanskiy-Poor-Verd\'{u} \cite{polyanskiy2010channel} for finite blocklength channel coding.
For lossy source coding and lossy joint source-channel coding with the probability of excess distortion as the loss criterion, the LP framework results in two levels of improvements on the asymptotically tight tilted-information based converses of Kostina and Verd\'{u} in \cite{kostina2012fixed} and \cite{kostina2013lossy}, respectively.

Fundamental to this framework is the observation that the finite blocklength coding problem can be posed equivalently as a nonconvex optimization problem over joint probability distributions. A natural optimizer's approach~\cite{kulkarni2014optimizer} to obtain lower bounds would then be via a convex relaxation of the nonconvex optimization problem. In particular, resorting to the ``lift-and-project'' technique due to Lovasz, Schrijver, Sherali, Adams and others~\cite{conforti2014integer}, we obtain a LP relaxation of the problem. From linear programming duality, we then get that the objective value of \textit{any feasible point of the dual of this LP relaxation yields a lower bound} on the optimal loss in the finite blocklength problem. As a result of this observation, the problem of obtaining converses reduces to constructing feasible points for the dual linear program.

The converses in~\cite{jose2016linear} stated above for various point-to-point settings emerge as special cases of this LP-based framework, implied by the construction of specific dual feasible points. This tightness of the LP relaxation shows that there is an alternative, asymptotically tight way of thinking about  optimal finite blocklength coding -- as the optimal packing of a pair of \textit{source} and \textit{channel flows} satisfying a certain error density bottleneck. The flows here are the variables of the dual program and the bottleneck, its constraint.

In this paper, we further this theme towards the Slepian-Wolf coding problem. In the present paper we observe that our LP relaxation has an operational interpretation based on optimal transport~\cite{villani2008optimal}, wherein one designs not only codes, but also couplings between them to minimize the resulting `error'.
Using the LP relaxation, we first establish new, clean, meta-converses in the point-to-point setting for lossy source-coding problems with side-information at the decoder; these converses are stronger than our earlier converses in~\cite{jose2016linear}, they imply the hypothesis testing and tilted information based converses of Kostina and Verd\'{u} \cite[Theorem 7,Theorem 8]{kostina2012fixed} and the converse of Palzer and Timo \cite[Theorem 1]{palzer2016converse},  and are, to the best of our knowledge, the strongest known.  Subsequently, we analyse the dual LP of the finite blocklength Slepian-Wolf coding problem. When extended to the networked Slepian-Wolf coding problem, the LP-based framework results in a large number of dual variables and constraints, which makes it quite challenging to analyse and interpret. Consequently, we devise an analytically simpler approach to construct feasible points of the dual program using feasible points of simpler point-to-point problems. This yields tight finite blocklength converses that improve on the hitherto best known converse for this problem, due to Miyake and Kanaya~\cite{miyake1995coding}. 

The dual variables of the LP relaxation of SW also have a structure of `source flows' and `channel flows'. Though, as yet, we do not have physical or operational interpretations for these `flows', they serve as useful analytical devices for synthesizing converse expressions for SW. We find that source and channel flows for problem SW follow a hierarchy such that flows at the highest level satisfy the error density bottleneck, whereas the flows at the next levels have to meet a bottleneck, dictated by the flows at the level above, along various paths in the network.
We show that the well-known information spectrum-based converse of Miyake and Kanaya \cite{miyake1995coding} results from a particular way of constructing these flows. Improvements on this converse follow by synthesizing these flows in a more sophisticated manner. Specifically, by synthesizing flows for the networked problem using flows from the following point-to-point problems: (a) lossless source coding of jointly encoded correlated sources $(S_1,S_2)$, (b) lossless source coding of $S_1$ with perfect side-information of $S_2$ available at the decoder, and (c) lossless source coding of $S_2$ with perfect side-information of $S_1$ at the decoder, 
we show that a new finite blocklength meta-converse results, which improves on the converse of Miyake and Kanaya.

The paper is organized as follows. In Section~\ref{sec:point-to-point}, we consider the point-to-point lossy source coding problem with side-information. By the LP framework and an appropriate choice of source and channel flows, we derive new tight finite blocklength converses for these problems. In Section~\ref{sec:LP}, we discuss the extension of the LP relaxation to problem SW and establish the duality based framework. In Section~\ref{sec:sythesize}, we illustrate how to synthesize  new finite blocklength converses for SW from point-to-point sub-problems and present a new finite blocklength converse which improves on the converse of Miyake and Kanaya.
Lastly, in Section~\ref{sec:interpretation}, we discuss the structure of the constraints of the dual program corresponding to SW and possible avenues for further strengthening of the bound.
\subsection{Notation}
Throughout this paper, we consider only discrete random variables. We make use of the following notation. Upper case letters $A,B$ represent random variables taking values in finite spaces  represented by calligraphic letters, $\Ascr,\Bscr$ respectively; lower case letters $a,b$ represent the specific values these random variables take. $\Ibb\{\cdot\}$ represents the indicator function which is equal to one when `$\cdot$' is true and is zero otherwise. 
$\Pscr(\cdot)$ denotes the set of all probability distributions on $`\cdot$' and $Q\in \Pscr(\cdot)$ represents a specific distribution. If $Q$ is a joint probability distribution, let $Q_{\bullet}$ denote the marginal distribution of `$\bullet$'. For example, $Q_{X|S}$ represents the vector with $Q_{X|S}(x|s)$ for  $x \in \mathcal{X},s\in \mathcal{S}$ as its components.
 Let $P_{A|B}P_{C|D}(a,b,c,d)$ stand for $P_{A|B}(a|b)P_{C|D}(c|d)$. If $P$ represents an optimization problem, then $\OPT(P)$ represents its optimal value and $\FEA(P)$ represents its feasible region.
LHS stands for Left Hand Side and RHS stands for Right Hand Side. The notation $a \independent b$ denotes that $a$ is independent of $b$.
   \section{Finite Blocklength Point-to-Point Source Coding}\label{sec:point-to-point}
  In this section, we consider the point-to-point lossy source coding problem and the lossless source coding problem with side information at the decoder. We employ the LP relaxation framework to obtain finite blocklength converses for these problems.  
  \subsection{Point-to-Point Lossy Source Coding}\label{sec:lossysource}
  \begin{figure}
\begin{center}
\includegraphics[scale=0.5,clip=true, trim = 1in 6in 0in 3.6in]{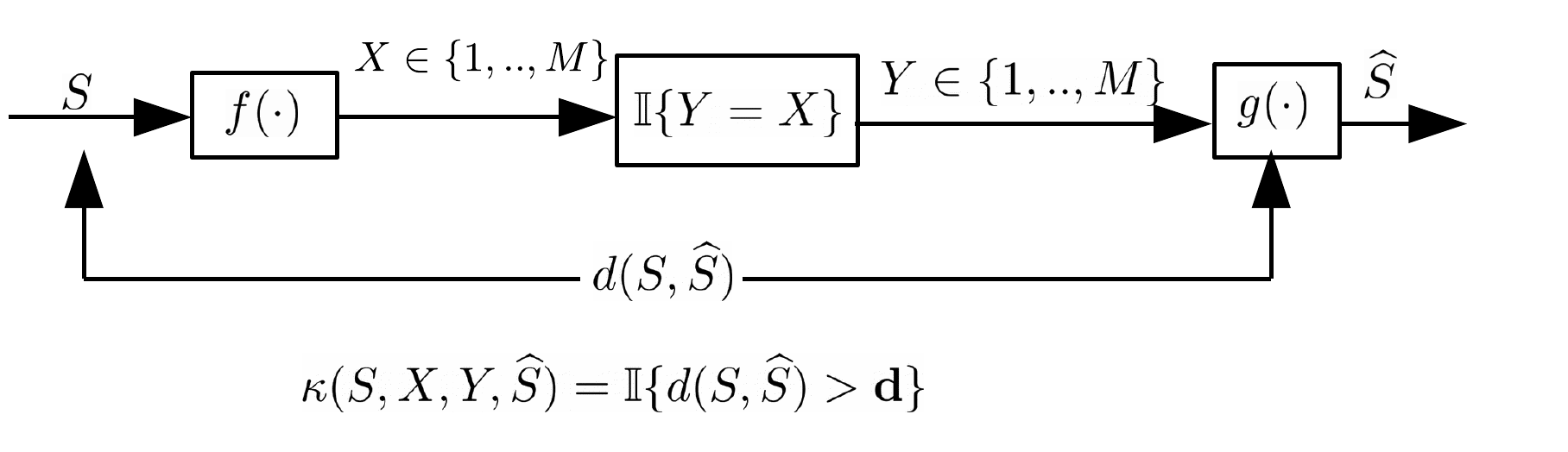} 
\caption{Lossy Source Coding}\label{fig:lossySC}
\vspace{-0.8cm}
\end{center}
\end{figure}
We begin with point-to-point lossy source coding. The finite blocklength lossy source coding problem (Figure~\ref{fig:lossySC}) with probability of excess distortion as the loss criterion can be posed as the following optimization problem,
 $$\problemsmalla{SC}
	{f,g}
	{\displaystyle \mathbb{E}[\Ibb\{d(S,\Shat)>\dbf\}]}
				 {\begin{array}{r@{\ }c@{\ }l}
				 				 X&=&f(S)\\
								 \widehat{S}&=&g(Y).
	\end{array}}
	$$ 
	Here, $S,X,Y$ and $\widehat{S}$ are discrete random variables taking values in fixed, finite spaces $\Sscr,\Xscr,\Yscr,\Sscrhat$ respectively, with $\mathcal{X}=\Yscr=\{1,\hdots,M\}$, $M \in \Nbb$ and $\Sscr=\Sscrhat$. $S$ represents the source message distributed according to a known distribution $P_S$. The source message is encoded according to $f:\Sscr \rightarrow \Xscr$ to get the signal $X$ which is transmitted across a deterministic channel with conditional probability distribution $P_{Y|X}= \Ibb\{Y=X\}$. $Y$ represents the channel output which is decoded according to $g:\Yscr \rightarrow \Sscrhat$ to get the message $\Shat $ at the destination. $d:\Sscr \times \Sscrhat \rightarrow [0,+\infty]$ represents the distortion measure and $\dbf\geq 0$ represents the distortion level. The optimization problem SC, thus, seeks to find a code $(f,g)$ which minimizes  $\mathbb{E}[\Ibb\{d(S,\Shat)>\dbf\}]=\Pbb[d(S,\Shat)>\dbf]$, 	the probability of excess distortion under the measure $\Pbb$ induced by $f,g$.
	
	 SC can be posed equivalently as the following optimization problem over joint probability distributions,
	$$
\problemsmalla{$\SC$}
	{Q,Q_{X|S},Q_{\widehat{S}|Y}}
	{\displaystyle \sum_{s,\shat}\Ibb\{d(s,\shat)>\dbf\} \sum_{x,y}Q(s,x,y,\shat)}
				 {\begin{array}{r@{\ }c@{\ }l}
				 Q(\zhat)&\equiv &P_{S} Q_{X|S} P_{Y|X}  Q_{\widehat{S}|Y}(\zhat),\\
				 				 \sum_{x} Q_{X|S}(x|s)&=& 1 \quad \forall s \in \mathcal{S},\\
\sum_{\shat} Q_{\widehat{S}|Y}(\shat|y)&=&1 \quad\forall y \in \mathcal{Y},\\
 Q_{X|S}(x|s) &\geq& 0\quad \forall s \in \mathcal{S},x \in \mathcal{X},\\
  Q_{\widehat{S}|Y}(\shat|y) &\geq& 0 \quad \forall \shat \in \mathcal{\widehat{S}},y \in \mathcal{Y},
	\end{array}}
	$$  
	where $\zhat:=(s,x,y,\shat) \in \widehat{\Zscr}$, $\widehat{\Zscr}:=\Sscr \times \Xscr \times \Yscr \times \Sscrhat$ and $P_S Q_{X|S} P_{Y|X}  Q_{\Shat|Y}(\zhat)\equiv$ $ P_S(s)Q_{X|S}(x|s)P_{Y|X}(y|x)Q_{\Shat|Y}(\shat|y)$. Here, $Q_{X|S}$ represents a randomized encoder and $Q_{\Shat|Y}$ represents a randomized decoder. We refer the readers to \cite{jose2016linear} for details on this formulation.
  
To obtain lower bounds on the optimal value of SC, we adopt the LP relaxation detailed in \cite{jose2016linear}. Towards this, we introduce a new variable $W(s,x,y,\shat)\equiv Q_{X|S}(x|s)Q_{\Shat|Y}(\shat|y)$ and obtain valid constraints involving $W$ through the constraints of the problem. Specifically, multiply both sides of the constraint $\sum_x Q_{X|S}(x|s)\equiv 1$ by $Q_{\Shat|Y}(\shat|y)$ for all $s,y,\shat$ and multiply both sides of $\sum_{\shat}Q_{\Shat|Y}(\shat|y)\equiv 1$ by $Q_{X|S}(x|s)$ for all $s,x,y$. Replacing the bilinear product terms in the resulting set of constraints and in the objective of $\SC$ with $W$, gives new linear  constraints in the variables $Q_{X|S},Q_{\Shat|Y},W$, which together with $Q_{X|S} \in \Pscr(\Xscr|\Sscr)$ and $Q_{\Shat|Y} \in \Pscr(\Sscrhat|\Yscr)$ give the following LP relaxation.
$$ \problemsmalla{LP}
	{ Q_{X|S},Q_{\widehat{S}|Y},W}
	{\displaystyle \sum_{\zhat}\Ibb\{d(s,\shat)>\dbf\}P_{S}(s)P_{Y|X}(y|x)W(\zhat)}
				 {\begin{array}{r@{\ }c@{\ }l}
				 				 \sum_{x} Q_{X|S}(x|s)&=& 1 \hspace{0.65cm} : \gamma^a(s) \qquad  \hspace{0.08cm} \forall s\\
\sum_{\shat} Q_{\Shat|Y}(\shat|y)&=&1 \hspace{0.65cm} :\gamma^b(y) \qquad \hspace{0.08cm}\forall y\\
 \sum_{x} W(\zhat)-Q_{\widehat{S}|Y}(\shat|y)&=&0  \hspace{0.65cm} :\lambda_{\ssf}(s,\shat,y)
  \hspace{0.1cm}  \forall s,\shat,y\\ 
 \sum_{\shat}W(\zhat)-Q_{X|S}(x|s)&=&0  \hspace{0.65cm} :\lambda_{\csf}(x,s,y)\hspace{0.1cm} \forall x,s,y\\
Q_{X|S},Q_{\Shat|Y},W &\geq& 0.
	\end{array}}  $$ 
	Above $\gamma^a,\gamma^b,\lambda_{\ssf}$ and $\lambda_{\csf}$ are Lagrange multipliers corresponding to the respective constraints.

\subsubsection{An operational interpretation via optimal transport}	
The above LP relaxation can be explained operationally by relating it to the optimal transport problem~\cite{villani2008optimal}. Note that for each $s \in \Sscr$ and $y \in \Yscr$, $W(s,\cdot,y,\cdot)$ is a \textit{coupling} on $\Xscr \times \Sscrhat$ between the marginals $Q_{X|S}(\cdot|s)$ and $Q_{\Shat|Y}(\cdot|y)$; let the set of such $W$ be denoted by $\Xi(Q_{X|S},Q_{\Shat|Y})$. The LP relaxation of SC is a \textit{nested} minimization --  the inner minimization is over all couplings $W\in \Xi(Q_{X|S},Q_{\Shat|Y})$ and the outer minimization is over all randomized codes $(Q_{X|S},Q_{\Shat|Y})$: 
\[\min_{Q_{X|S},Q_{\Shat|Y}} \min_{W \in \Xi(Q_{X|S},Q_{\Shat|Y})} \sum_{\zhat}\Ibb\{d(s,\shat)>\dbf\} P_{S}(s)P_{Y|X}(y|x)W(\zhat).\] 
The original problem SC has the outer minimization over codes, but in place of the inner minimization over $W$ it employs the product $Q_{X|S}(\cdot|\cdot)Q_{\Shat|Y}(\cdot|\cdot) \in \Xi(Q_{X|S},Q_{\Shat|Y})$ to obtain the distribution. 
Thus the LP relaxation is arrived at by considering the term $Q_{X|S}(\cdot|\cdot)Q_{\Shat|Y}(\cdot|\cdot)$ in SC as an  element of $\Xi(Q_{X|S},Q_{\Shat|Y})$ and minimizing the resulting cost over all elements of $\Xi(Q_{X|S},Q_{\Shat|Y})$. Operationally speaking,  the LP relaxation seeks to design codes \textit{and} couplings between them that minimize the error under the joint distribution induced by the coupling.

We caution the readers that for multiterminal problems, one must apply this interpretation with additional caveats. We discuss this in Section~\ref{sec:opttransSW}.

\subsubsection{Duality and bounds}	
 Employing the Lagrange multipliers corresponding to the constraints of LP, we obtain the following dual of LP,
$$
\maxproblemsmall{DP}
	{\gamma^a,\gamma^b,\lambda_{\ssf},\lambda_{\csf}}
	{\displaystyle \sum_{s}\gamma^a(s)+\sum_{y}\gamma^b(y)}
				 {\begin{array}{r@{\ }c@{\ }l}
				 				 \gamma^a(s)- \sum_y \lambda_{\csf}(x,s,y) &\leq& 0 \quad \hspace{0.8cm} \forall x,s\hspace{0.2cm}({\rm P1})\\
\gamma^b(y)- \sum_s \lambda_{\ssf}(s,\shat,y) &\leq& 0  \quad \hspace{0.8cm} \forall \shat,y\hspace{0.21cm}({\rm P2})\\
\lambda_{\ssf}(s,\shat,y)+\lambda_{\csf}(x,s,y) &\leq& \Sigma(\zhat)\hspace{0.65cm}\forall \zhat\hspace{0.53cm}({\rm P3})
 	\end{array}}
	$$ where $\Sigma(\zhat)= \Ibb\{d(s,\shat)>\dbf\} P_S(s) \Ibb\{y=x\}$ for all $\zhat$, since $P_{Y|X}(y|x)=\I{y=x}$.
	
	 In problem DP, it is optimal to choose $\gamma^a(s)$ and $\gamma^b(y)$ such that (P1) and (P2) hold with equality, \ie, $\gamma^a(s) \equiv \min_x \sum_y \lambda_{\csf}(x,s,y)$ and $\gamma^b(y)\equiv \min_{\shat} \sum_s \lambda_{\ssf}(s,\shat,y)$. Then 
the optimal value of DP with $\Sigma(\zhat)$ as the RHS of (P3) evaluates to,
\begin{align}
&\OPT(\DP,\Sigma(\zhat))\non \\&=\max_{\lambda_{\ssf},\lambda_{\csf}} \biggl \lbrace \sum_s \min_x \sum_y \lambda_{\csf}(x,s,y)+\sum_y \min_{\shat}\sum_s \lambda_{\ssf}(s,\shat,y)\biggr \rbrace\non\\
& \quad\mbox{s.t} \hspace{1.3cm} \lambda_{\csf}(x,s,y)+\lambda_{\ssf}(s,\shat,y)\leq \Sigma(\zhat)\ \ \forall \zhat \in \widehat{\Zscr}.\label{eq:optDP}
\end{align}
It follows that if we construct functions $\lambda_{\ssf}:\Sscr \times \Sscrhat \times \Yscr \rightarrow \Real$ and $\lambda_{\csf}: \Sscr \times \Xscr \times \Yscr \rightarrow \Real$ satisfying \eqref{eq:optDP}, then linear programming duality implies the following lower bound on $\OPT(\SC)$,
\begin{align}
&\OPT(\SC)\geq \OPT(\LP) =\OPT(\DP) \non\\&\geq \sum_s \min_x \sum_y \lambda_{\csf}(x,s,y)+\sum_y \min_{\shat}\sum_s \lambda_{\ssf}(s,\shat,y).
\label{eq:ptpframework}
\end{align}

Notice that $\lambda_{\ssf}$ and $\lambda_{\csf}$ are functions on subspaces of $\Sscr \times \Xscr \times \Yscr \times \Sscrhat$.
$\lambda_{\csf}$ is a function of the source signal $s$, the channel input $x$ and channel output $y$; we call this function a  \textit{channel flow}. On the other hand, $\lambda_{\ssf}$ is a function of the source signal $s$, the decoder input $y$ and decoder output $\shat$. We refer to it as a \textit{source flow}.
Hence, for the point-to-point finite blocklength source coding problem, our LP-based framework reduces to constructing a \textit{source flow} and a \textit{channel flow} such that they satisfy the bottleneck imposed by the constraint (P3). The RHS of (P3) is the ``error density'', $\Sigma(\zhat)=P_S(s)\Ibb\{y=x\}\Ibb\{d(s,\shat)>\dbf\}$,
 and hence the challenge is to \textit{optimally pack} a source flow and a channel flow so as to not exceed the error density.

It was shown in \cite{jose2016linear} that by an appropriate construction of these source and channel flows, a new finite blocklength converse for lossy source coding results which improves on the tilted information based converse of Kostina and Verd\'{u} \cite{kostina2012fixed}. Modifying and generalizing this construction of flows, we now present a new metaconverse for lossy source coding, which \textit{implies} our improvement on the Kostina-Verd\'{u} converse and the hypothesis testing based converse in \cite[Theorem 8]{kostina2012fixed}.
To the best of our knowledge, the metaconverse below is the strongest known. 
\begin{theorem}[Metaconverse for Lossy Source Coding]\label{thm:metaconverse}
Consider problem SC. For any code,
\begin{align}
&\Ebb[\Ibb\{d(S,\Shat)>\dbf\}]\geq \OPT(\SC)\geq \OPT(\LP)=\OPT(\DP)\non\\& \geq \sup_{0\leq \phi(s) \leq P_S(s)}\biggl\lbrace \sum_s \phi(s)-M\max_{\shat}\sum_s \phi(s) \Ibb\{d(s,\shat)\leq \dbf\}\biggr \rbrace \label{eq:metalossy},
\end{align}
where the supremum is over all functions $\phi: \Sscr \rightarrow [0, 1]$ such that $0 \leq \phi(s)\leq P_S(s)$ for all $s \in \Sscr$.
\end{theorem}
\begin{proof}
Consider the following values of source and channel flows,
\begin{align}
\lambda_{\csf}(x,s,y)&\equiv \Ibb\{y=x\}\phi(s),\label{eq:jointfeasible}\\
\lambda_{\ssf}(s,\shat,y)&\equiv -\phi(s)\Ibb\{d(s,\shat)\leq \dbf\}.\non
\end{align} We now check if the above choice of flows satisfy constraint (P3). For this, consider the following two cases.\\
Case 1: $\Ibb\{d(s,\shat)> \dbf\}=1$.\\
In this case, $\lambda_{\ssf}(s,\shat,y)=0$ and $\lambda_{\csf}(x,s,y)=\Ibb\{y=x\}\phi(s)\leq P_S(s)\Ibb\{y=x\},$ which is the RHS of (P3).
\\Case 2: $\Ibb\{d(s,\shat)> \dbf\}=0$.\\In this case, the RHS of (P3) is zero and LHS becomes,
$\Ibb\{y=x\}\phi(s)-\phi(s) \leq 0,$ thereby satisfying (P3).
Hence, the considered choice of flows satisfy constraint (P3).
 Consequently, the required lower bound follows from \eqref{eq:ptpframework}
by taking supremum over $\phi$ such that $0\leq \phi(s)\leq P_S(s)$ $ \forall s \in \Sscr$.
\end{proof}
In particular, choosing $\phi(s)=\min\{P_{S}(s),z(s)\}$ in \eqref{eq:metalossy} where $z:\Sscr \rightarrow [0,\infty)$, and taking the  supremum over such $z$, we get the following bound,
\begin{align}
&\Ebb[\Ibb\{d(S,\Shat)>\dbf\}]\geq \OPT(\SC)\geq \OPT(\LP)\non\\& \geq \sup_{z\geq 0}\biggl\lbrace \sum_s \min\{P_S(s),z(s)\}\non\\&-M\max_{\shat}\sum_s \min\{P_S(s),z(s)\} \Ibb\{d(s,\shat)\leq \dbf\}\biggr \rbrace.\label{eq:metalossy1}
\end{align}
\begin{remarkc}[(Choice of Flows)] An easy way of motivating the choice of flows is as follows. 
Observe that if $0\leq \phi(s) \leq P_S(s)$ for all $s,$ we have that,
\begin{align*}
\Sigma(\zhat)&\geq \phi(s)\Ibb\{y=x\}\Ibb\{d(s,\shat)>\dbf\}\\&
=\phi(s)\Ibb\{y=x\}-\phi(s)\Ibb\{y=x\}\Ibb\{d(s,\shat)\leq \dbf\}. 
\end{align*}
An obvious choice of the flows would thus be $\lambda_{\csf}(x,s,y)=\phi(s)\Ibb\{y=x\}$ and $\lambda_{\ssf}(s,\shat,y)\leq -\phi(s)\Ibb\{y=x\}\Ibb\{d(s,\shat)\leq \dbf\}$, which results in our metaconverse in \eqref{eq:metalossy}.
\end{remarkc}

The following results are corollaries to the metaconverse. $j_S(\cdot,\dbf)$ below is the $d$-tilted information; we refer the reader to \cite{kostina2012fixed} for details.
\begin{corollary}
\textit{(Metaconverse Recovers Improvement on Kostina-Verd\'{u} Converse from \cite[Corollary 5.9]{jose2016linear})}\\
The following converse follows from \eqref{eq:metalossy1}:
\begin{align*}
&\Ebb[\Ibb\{d(S,\Shat)>\dbf\}]\geq \OPT(\SC)\geq \OPT(\LP)=\OPT(\DP)\non\\& \geq \sup_{\gamma} \biggl \lbrace \Pbb[j_S(S,\dbf) \geq \gamma+ \log M]+\frac{1}{M}\times\nonumber\\&\quad \sum_sP_S(s)\exp(j_S(s,\dbf)-\gamma)\Ibb{\{j_S(s,\dbf) < \log M+\gamma\}}\non\\&-\max_{\shat}\exp(-\gamma)\sum_s P_S(s)\exp(j_S(s,\dbf))\Ibb\{d(s,\shat)\leq \dbf\} \biggr \rbrace,
\end{align*}which is the improvement on the Kostina-Verd\'{u} tilted information based converse in \cite[Corollary 5.9]{jose2016linear}.
\end{corollary}
\begin{proof}
To see this, take $z(s)=P_S(s)\frac{1}{M}\exp(j_S(s,\dbf)-\gamma)$ for any scalar $\gamma$ and lower bound $-M\sup_{\shat}\sum_s \min\{P_S(s),z(s)\} \Ibb\{d(s,\shat)\leq \dbf\}$ with $-M\sup_{\shat}\sum_s z(s) \Ibb\{d(s,\shat)\leq \dbf\}$. Subsequently, take supremum over $\gamma$ to get the required bound.
\end{proof}

Consider a binary hypothesis testing problem between distributions $P_S $ and $ Q_S$. Let $\alpha_{\theta}(P_S,Q_S)$ represent the minimum type-I error, $\sum_s P_S(s)T(s)$ over all tests $T$ such that the type-II error, $\sum_s Q_S(s) (1-T(s))$ is at most $\theta$. The following corollary shows that the metaconverse in Theorem~\ref{thm:metaconverse} in fact recovers the hypothesis testing based converse of Kostina and Verd\'{u} \cite[Theorem 8]{kostina2012fixed}.
\begin{corollary}\textit{(Metaconverse Recovers Kostina-Verd\'{u} Hypothesis testing based Converse from \cite[Theorem 8]{kostina2012fixed})}\\
The following converse follows from \eqref{eq:metalossy1},
\begin{align}
&\Ebb[\Ibb\{d(S,\Shat)>\dbf\}]\geq \OPT(\SC)\geq \OPT(\LP)=\OPT(\DP)\non\\& \stackrel{(a)}{\geq} \sup_{Q_S \in \Pscr(\Sscr)}\sup_{\beta\geq 0}\biggl \lbrace \sum_s \min\{P_S(s), \beta Q_S(s)\}-\beta M^{*} \biggr \rbrace \non\\
&  \stackrel{(b)}{ =}\sup_{Q_S \in \Pscr(\Sscr)} \alpha_{M^{*}}(P_S,Q_S)\label{eq:lossyhypoth},
\end{align} which is equivalent to the hypothesis testing based converse of Kostina and Verd\'{u} in \cite[Theorem 8]{kostina2012fixed}. Here, $M^{*}=M\max_{\shat}\mathbb{Q}[d(S,\shat)\leq \dbf]$ and $\mathbb{Q}$ is the measure on $\Sscr$ induced by $Q_S$.
\end{corollary}
\begin{proof}
 To recover the converse in $(a)$ from \eqref{eq:metalossy1}, take $z(s)=\beta Q_S(s)$ where $\beta\geq 0$ and lower bound $-M\sup_{\shat}\sum_s \min\{P_S(s),z(s)\} \Ibb\{d(s,\shat)\leq \dbf\}$ with $-M\sup_{\shat}\sum_s z(s) \Ibb\{d(s,\shat)\leq \dbf\}=-\beta M^*$. Subsequently, take the supremum over $\beta \geq 0$ and $Q_S \in \Pscr(\Sscr)$ to get the required bound. The proof of the relation in $(b)$ is included in Corollary~\ref{cor:alphahypoth} in Appendix~A.
\end{proof}

\begin{corollary}\textit{(Metaconverse Recovers Palzer-Timo Converse \cite[Theorem 1]{palzer2016converse})}\\
The following converse follows from \eqref{eq:metalossy1},
\begin{align*}
&\Ebb[\Ibb\{d(S,\Shat)>\dbf\}]\geq \OPT(\SC)\geq \OPT(\LP)=\OPT(\DP)\non\\& \geq \sup_{\beta \in \Real}\biggl \lbrace \Pbb[j_S(S,\dbf)\geq \beta]\\&-M\max_{\shat}\Pbb[j_S(S,\dbf)\geq \beta,d(S,\shat)\leq \dbf]\biggr \rbrace,
\end{align*} which is the converse of Palzer and Timo in \cite[Theorem 1]{palzer2016converse}.
\end{corollary}
\begin{proof}
To obtain this converse from \eqref{eq:metalossy1}, take $z(s)=P_S(s)\Ibb\{j_S(s,\dbf)\geq \beta\}$ where $\beta\geq 0$ in  \eqref{eq:metalossy} and take supremum over $\beta\geq 0$. Notice that in this case, $\min\{z(s),P_S(s)\}=z(s)$ since $z(s)=P_S(s)\Ibb\{j_S(s,\dbf)\geq \beta\}\leq P_S(s)$.
\end{proof}

The finite blocklength lossless data compression problem results from SC by setting $d(S,\Shat)=\Ibb\{S \neq \Shat\}$ and $\dbf=0$. The following corollary particularizes the metaconverse  in \eqref{eq:metalossy} to lossless data compression. In this case, the metaconverse takes a particularly simple form.
\begin{corollary}[Metaconverse for Lossless Source Coding]\label{cor:metalossless}
For lossless data compression, consider the setting of Theorem~\ref{thm:metaconverse} with $d(S,\Shat)=\Ibb\{S \neq \Shat\}$ and $\dbf=0$. Consequently, for any code, the following converse follows from Theorem~\ref{thm:metaconverse},
\begin{align}
&\Ebb[\Ibb\{S \neq \Shat\}]\geq \OPT(\SC)\geq \OPT(\LP)= \OPT(\DP)\non\\&\geq  \sup_{0 \leq \phi\leq P_S}\biggl\lbrace \norm{\phi}_1-M\norm{\phi}_\infty \biggr \rbrace. \label{eq:hyplossless}
\end{align}
\end{corollary}
In the above converse we have viewed $\phi$ as a vector in $\Real^{|\Sscr|}$ that is nonnegative and dominated by $P_S \in \Real^{|\Sscr|}.$ The maximization in \eqref{eq:hyplossless} is a tradeoff between increasing the $\ell_1$ norm of $\phi$ on the one hand, and decreasing the $\ell_\infty$ norm of $\phi$ on the other. One plausible strategy for this tradeoff is to take $\phi(s) = P_S(s)$ for those $s \in \Sscr$ for which $P_S(s)$ is not too large, and zero otherwise. Specifically, one may take $\phi(s) = P_S(s) \I{P_S(s) \leq  \frac{\exp(-\gamma)}{M} }$ for some $\gamma\geq 0.$ Then the RHS of \eqref{eq:hyplossless} is lower bounded by
\[\sup_{\gamma\geq 0} \big\{\Pbb[ h(S) \geq \log M +\gamma] - \exp(-\gamma)\big \}, \]
where $h(S) = -\log P_S(S).$ The above converse is~\cite[Theorem~7]{kostina2012fixed} specialized to the lossless case.
 

Having outlined the LP based framework for point-to-point lossless source coding, we now consider three problems that will serve as sub-problems for analysing problem SW.
\subsubsection{Lossless Coding of Jointly Encoded Correlated Sources $(S_1,S_2)$}
In this sub-problem of Slepian-Wolf coding problem, the correlated sources
$S_1,S_2$ are \textit{jointly encoded} by $f :\Sscr_1\times \Sscr_2 \rightarrow \Xscr_1 \times \Xscr_2$ to get $(X_1,X_2)$. $(X_1,X_2)$ is sent through the channel $P_{Y_1,Y_2|X_1,X_2}=\Ibb\{(Y_1,Y_2)=(X_1,X_2)\}$ to get $(Y_1,Y_2)$ which is then decoded according to $g :\Yscr_1 \times \Yscr_2 \rightarrow \Sscrhat_1\times\Sscrhat_2$.	
The objective, as for SW problem, is to losslessly recover $(S_1,S_2)$ at the destination. 

It is easy to see that the above joint encoding problem is equivalent to the point-to-point lossless source coding problem SC  with $S:=(S_1,S_2)$, $X:=(X_1,X_2)$, $Y:=(Y_1,Y_2)$, $\Shat:=(\Shat_1,\Shat_2)$ and $d(S,\Shat)=\Ibb\{S \neq \Shat\}$ with $\dbf=0$. 
Consequently, to obtain finite blocklength converses for the joint encoding problem of correlated sources,
%
 we resort to the following generalized version of DP for lossless source coding problem,
$$
\begin{small}
\maxproblemsmall{DPJE}
	{\gammah^a,\gammah^b,\lambdah_{\ssf},\lambdah_{\csf}}
	{\displaystyle \sum_{s_1,s_2}\gammah^a(s_1,s_2)+\sum_{y_1,y_2}\gammah^b(y_1,y_2)}
				 {\begin{array}{r@{\ }c@{\ }l}
				 				 \gammah^a(s_1,s_2)- \sum_{y_1,y_2} \lambdah_{\csf}(s_1,s_2,x_1,x_2,y_1,y_2) &\leq& 0 \\   \forall x_1,x_2,s_1,s_2 &({\rm A1})&\\
\gammah^b(y_1,y_2)- \sum_{s_1,s_2} \lambdah_{\ssf}(s_1,s_2,\shat_1,\shat_2,y_1,y_2) &\leq& 0  \\  \forall \shat_1,\shat_2,y_1,y_2 &({\rm A2})&\\
\lambdah_{\ssf}(s_1,s_2,\shat_1,\shat_2,y_1,y_2)+\lambdah_{\csf}(s_1,s_2,x_1,x_2,y_1,y_2) &\leq& \hspace{-0.3cm}\Upsilon(\ztilde)\\ \forall \ztilde &({\rm A3})&\\
 	\end{array}}
\end{small}	
$$ where $\ztilde:=(s_1,s_2,x_1,x_2,y_1,y_2,\shat_1,\shat_2)$, $\Upsilon(\tilde{z})$ $= \Ibb\{(s_1,s_2)\neq (\shat_1,\shat_2)\}$ $ P_{S_1,S_2}(s_1,s_2) \Ibb\{(y_1,y_2)=(x_1,x_2)\}$ for all $\ztilde$. Though this is a straightforward generalization of DP, we will need this later and hence we have included it here.

As in the case of DP, taking $\gammah^a$ and $\gammah^b$ such that (A1) and (A2) hold with equality,
  $\OPT(\DPJE)$ can be written in terms of the channel flow $\lambdah_{\csf}(s_1,s_2,x_1,x_2,y_1,y_2)$ and source flow $\lambdah_{\ssf}(s_1,s_2,\shat_1,\shat_2,y_1,y_2).$ 
The metaconverse for lossless source coding problem  in Corollary~\ref{cor:metalossless} then readily implies the following corollary.
\begin{corollary}[Metaconverse for Jointly Encoded Sources]\label{cor:metaJE}
Consider problem SC  with $S:=(S_1,S_2)$, $X:=(X_1,X_2)$, $Y:=(Y_1,Y_2)$, $\Shat:=(\Shat_1,\Shat_2)$ and $d(S,\Shat)=\Ibb\{S \neq \Shat\}$ with $\dbf=0$. Consequently for any code, we have from Corollary~\ref{cor:metalossless},
\begin{align}
&\Ebb[\Ibb\{(S_1,S_2) \neq (\Shat_1,\Shat_2)\}]\geq  \OPT(\DPJE) \non\\& \geq\sup_{ 0\leq \phih(s_1,s_2)\leq P_{S_1,S_2}(s_1,s_2)}\biggl\lbrace \sum_{s_1,s_2} \phih(s_1,s_2)\non \\&\qquad
-M_1M_2\max_{\shat_1,\shat_2} \phih(\shat_1,\shat_2) \biggr \rbrace, \label{eq:JEmetaconverse}
\end{align}where the supremum is over  $\phih:\Sscr_1 \times \Sscr_2 \rightarrow [0,1]$ such that $0\leq \phih(s_1,s_2)\leq P_{S_1,S_2}(s_1,s_2)$ for all $s_1 \in \Sscr_1,s_2 \in \Sscr_2$.
\end{corollary}
\begin{proof}
To obtain the required converse, we consider the following choice for the flows in DPJE, which generalizes  the ones adopted in \eqref{eq:jointfeasible}.
\begin{align}
\lambdah_{\csf}(x_1,x_2,s_1,s_2,y_1,y_2)\hspace{-0.03cm}&\hspace{-0.1cm}\equiv \Ibb\{(y_1,y_2)\hspace{-0.05cm}=\hspace{-0.05cm}(x_1,x_2)\}\phih(s_1,s_2)\hspace{-0.1cm}\label{eq:JEmetafeasible}\\
\lambdah_{\ssf}(s_1,s_2,\shat_1,\shat_2,y_1,y_2)&\equiv -\phih(s_1,s_2)\Ibb\{(s_1,s_2)= (\shat_1,\shat_2)\},\non
\end{align} The feasibility of these flows with respect to (A3) can be verified as in the proof of Theorem~\ref{thm:metaconverse} and we skip the proof here.
\end{proof}
\subsection{Lossless Source Coding of $S_1$ with $S_2$ as the Side-Information}
We now consider the following sub-problem of Slepian-Wolf coding: $S_1$ is to be recovered losslessly at the destination with $S_2$ available as  side-information at the decoder (Figure~\ref{Fig:sideinformation}).
  Towards this, $S_1$ is encoded according to $f_1 :\Sscr_1 \rightarrow \Xscr_1$ to get $X_1$, which is transmitted through the channel $P_{Y_1|X_1}=\Ibb\{Y_1=X_1\}$ to get $Y_1$. $S_2$ is the side information available at the decoder which decodes according to $g :\Sscr_2 \times \Yscr_1 \rightarrow \Sscrhat_1$ to get $\Shat_1$.
 The finite blocklength source coding of $S_1$ given $S_2$ as the side information can be then posed as the following optimization problem,
  $$\problemsmalla{$\SID_{1|2}$}
	{f_1,g}
	{\displaystyle \mathbb{E}[\Ibb\{S_1 \neq \Shat_1\}]}
				 {\begin{array}{r@{\ }c@{\ }l}
				 				 X_1&=&f_1(S_1),\quad
								 \Shat_1=g(S_2,Y_1).
	\end{array}}
	$$ Thus, $\SID_{1|2}$ seeks to obtain a code $(f_1,g)$ which minimizes $\mathbb{E}[\Ibb\{S_1 \neq \Shat_1\}]=\Pbb[S_1 \neq \Shat_1]$, the average probability of error. 
	

To obtain finite blocklength converses, we employ the LP relaxation approach in \cite{jose2016linear} to obtain the following LP relaxation of the problem $\SID_{1|2}$. 
$$
\begin{small}
 \problemsmalla{$\LPSI_{1|2}$}
	{ Q_{X_1|S_1},Q_{\widehat{S}_1|Y_1,S_2},W}
	{\displaystyle \sum_{\zbar}\Psi(\zbar)W(\zbar)}
				 {\begin{array}{r@{\ }c@{\ }l}
				 				 \sum_{x_1} Q_{X_1|S_1}(x_1|s_1)&\equiv& 1 \hspace{0.01cm} : \gammab^a(s_1)\\
\sum_{\shat_1} Q_{\Shat_1|Y_1,S_2}(\shat_1|y_1,s_2)&\equiv &1 \hspace{0.01cm} :\gammab^b(y_1,s_2) \\
 \sum_{x_1} W(\zbar)-Q_{\widehat{S}_1|Y_1,S_2}(\shat_1|y_1,s_2)&=&0   :\lambdab^{(1|2)}_{\ssf}(\sbf,\shat_1,y_1)
 \\ 
 \sum_{\shat_1}W(\zbar)-Q_{X_1|S_1}(x_1|s_1)&\equiv&0   :\lambdab^{(1|2)}_{\csf}(\sbf,x_1,y_1)\\
Q_{X_1|S_1},Q_{\Shat_1|Y_1,S_2},W &\geq& 0,
	\end{array}}
	 \end{small} 
	 $$ where $ \sbf:=(s_1,s_2),$ $\zbar:=(x_1,s_1,s_2,\shat_1,y_1)$ and
	$\Psi(\zbar)=P_{S_1,S_2}(s_1,s_2)\Ibb\{s_1\neq \shat_1\}\Ibb\{y_1=x_1\}.$
	
	Employing the Lagrange multipliers $\gammab^a,\gammab^b,\lambdab^{(1|2)}_{\ssf}$ and $\lambdab^{(1|2)}_{\csf}$ corresponding to the constraints of $\LPSI_{1|2}$, we obtain the following dual of $\LPSI_{1|2}$.
	$$\maxproblemsmall{$\DPSI_{1|2}$\hspace{0.17cm}}
	{ \gammab^a,\gammab^b,\lambdab^{(1|2)}_{\ssf},\lambdab^{(1|2)}_{\csf}b}
	{\displaystyle \sum_{s_1}\gammab^a(s_1)+\sum_{y_1,s_2} \gammab^b(y_1,s_2)}
				 {\begin{array}{r@{\ }c@{\ }l}
				 				 				\gammab^a(s_1)-\sum_{y_1,s_2} \lambdab^{(1|2)}_{\csf}(s_1,s_2,x_1,y_1) &\leq & 0 \\ \forall x_1,s_1\quad ({\rm B1})&  &\\
				 				 				\gammab^b(s_2,y_1)-\sum_{s_1} \lambdab^{(1|2)}_{\ssf}(s_1,s_2,\shat_1,y_1) &\leq & 0\\ \forall s_2,\shat_1,y_1\quad ({\rm B2})&  &\\	\lambdab^{(1|2)}_{\ssf}(s_1,s_2,\shat_1,y_1)+\lambdab^{(1|2)}_{\csf}(s_1,s_2,x_1,y_1)&\leq & \hspace{-0.15cm} \Psi(\zbar)   \\  \quad \forall \zbar \quad ({\rm B3})& &
\end{array}} $$ 
Choosing $\gammab^a(s_1)$ and $\gammab^b(s_2,y_1)$ such that (B1) and (B2) hold with equality, 
$\OPT(\DPSI_{1|2})$ can be written in terms of $\lambdab^{(1|2)}_{\csf}(s_1,s_2,x_1,y_1)$ and $\lambdab^{(1|2)}_{\ssf}(s_1,s_2,\shat_1,y_1).$
Notice that $\lambdab^{(1|2)}_{\csf}$ is a function of $x_1,s_1,s_2,y_1$. Thus, for each $s_2 \in \Sscr_2$, it is akin to a channel flow of the point-to-point source coding problem with $S_1$ as the source. Likewise, for each $s_2 \in \Sscr_2$, $\lambdab^{(1|2)}_{\ssf}(s_1,s_2,\shat_1,y_1)$ is akin to a \textit{source flow} for this problem.
Following these observations, we now show that an appropriate construction of these source and channel flows results in the following finite blocklength converse for $\SID_{1|2}$.
\begin{theorem}\label{thm:metaside12}
Consider the problem $\SID_{1|2}$. For any code, the following lower bound holds, 
\begin{align} 
&\Ebb[\Ibb\{S_1\neq \Shat_1\}]\geq \OPT(\SID_{1|2})\geq \non\\
& \OPT(\LPSI_{1|2}) \geq \sup_{0\leq \phi^{(1|2)} \leq P_{S_1,S_2}}\hspace{-0.1cm}\biggl \lbrace \hspace{-0.05cm}\sum_{s_1,s_2}\phi^{(1|2)}(s_1,s_2)\non \\& \qquad -M_1 \sum_{s_2}\max_{\shat_1}\phi^{(1|2)}(\shat_1,s_2) \biggr\rbrace, \label{eq:SIDmetaconverse}
\end{align}where the supremum is over  $\phi^{(1|2)}:\Sscr_1 \times \Sscr_2 \rightarrow [0,1]$ such that $\phi^{(1|2)}(s_1,s_2)\leq P_{S_1,S_2}(s_1,s_2)$ for all $s_1 \in \Sscr_1, s_2 \in \Sscr_2$.
\end{theorem}
\begin{proof}
To obtain the required converse, we consider the following values for the source flow and channel flow,
\begin{align}
\lambdab^{(1|2)}_{\csf}(s_1,s_2,x_1,y_1)&=\Ibb\{x_1=y_1\} \phi^{(1|2)}(s_1,s_2),\non\\
\lambdab^{(1|2)}_{\ssf}(s_1,s_2,\shat_1,y_1)&=-\phi^{(1|2)}(s_1,s_2)\Ibb\{s_1=\shat_1\}\label{eq:SIDfeasible}.
\end{align} The feasibility of these flows with respect to (B3) can be verified as in the proof of Theorem~\ref{thm:metaconverse}. Consequently, employing linear programming duality and taking supremum over $\phi^{(1|2)}$ gives the required bound.
\end{proof}
Notice that choosing $\phi^{(1|2)}(s_1,s_2)=\min\{P_{S_1,S_2}(s_1,s_2),\eta_2(s_1,s_2)\}$, $\eta_2 :\Sscr_1 \times \Sscr_2 \rightarrow [0,1]$, yields the following bound,
$\Ebb[\Ibb\{S_1\neq \Shat_1\}]\geq \OPT(\SID_{1|2})\geq $
\begin{align}
& \OPT(\LPSI_{1|2})\geq \sup_{\eta_2 \geq 0}\biggl \lbrace \sum_{s_1,s_2}\min\{P_{S_1,S_2}(s_1,s_2),\eta_2(s_1,s_2)\}\non\\&-M_1 \sum_{s_2}\max_{\shat_1}\min\{P_{S_1,S_2}(s_1,s_2),\eta_2(s_1,s_2)\} \biggr\rbrace. \label{eq:SIDmetaconverse1}
\end{align}
When particularized to $\eta_2(s_1,s_2)=P_{S_2}(s_2)\frac{\exp(-\beta)}{M_1}$ for $\beta\geq0$, where $P_{S_2}(s_2)=\sum_{s_1}P_{S_1,S_2}(s_1,s_2)$ and taking supremum over $\beta$, the converse in \eqref{eq:SIDmetaconverse1} becomes,
\begin{align}
&\Ebb[\Ibb\{S_1\neq \Shat_1\}]\geq  \OPT(\LPSI_{1|2})\geq \OPT(\DPSI_{1|2})\non\\
&\stackrel{(a)}{\geq}\sup_{\beta\geq 0}\biggl \lbrace \Pbb[h_{S_1|S_2}(S_1|S_2)\geq \log M_1+\beta]\non\\&+\frac{\exp(-\beta)}{M_1}\sum_{s_1,s_2}P_{S_2}(s_2)\Ibb\{ h_{S_1|S_2}(s_1|s_2)<\log M_1+\beta\}\non\\&-M_1 \sum_{s_2}\sup_{\shat_1}\min \biggl \lbrace P_{S_1,S_2}(\shat_1,s_2),P(s_2)\frac{\exp(-\beta)}{M_1}\biggr \rbrace \biggr\rbrace \label{eq:improvedSIDconverse}\\
&\stackrel{(b)}{\geq}\sup_{\beta\geq 0}\biggl \lbrace \Pbb[h_{S_1|S_2}(S_1|S_2)\geq \log M_1+\beta]-\exp(-\beta)\biggr \rbrace. \hspace{-0.3cm}\label{eq:SIDconverse}
\end{align}
Here $h_{A|B}(a|b)\equiv -\log P_{A|B}(a|b)$ is the conditional entropy density. 
The inequality in $(a)$ follows from the definition of conditional entropy density. The inequality in $(b)$ follows by lower bounding the non-negative term corresponding to $\Ibb\{ h_{S_1|S_2}(s_1|s_2)<\log M_1+\beta\}$ in \eqref{eq:improvedSIDconverse} by zero and upper bounding $\min\{P_{S_1,S_2}(\shat_1,s_2),P_{S_2}(s_2)\frac{\exp(-\beta)}{M_1}\}$ with $P_{S_2}(s_2)\frac{\exp(-\beta)}{M_1}$.
Notice that the converse in \eqref{eq:SIDconverse} is the well-known converse for lossless source-coding problem with side-information at the decoder. The converse in \eqref{eq:improvedSIDconverse} provides a new improvement on the standard converse.
\begin{figure}
\centering
\includegraphics[scale=0.45,clip=true,trim=0in 6.5in 0in 3.3in]{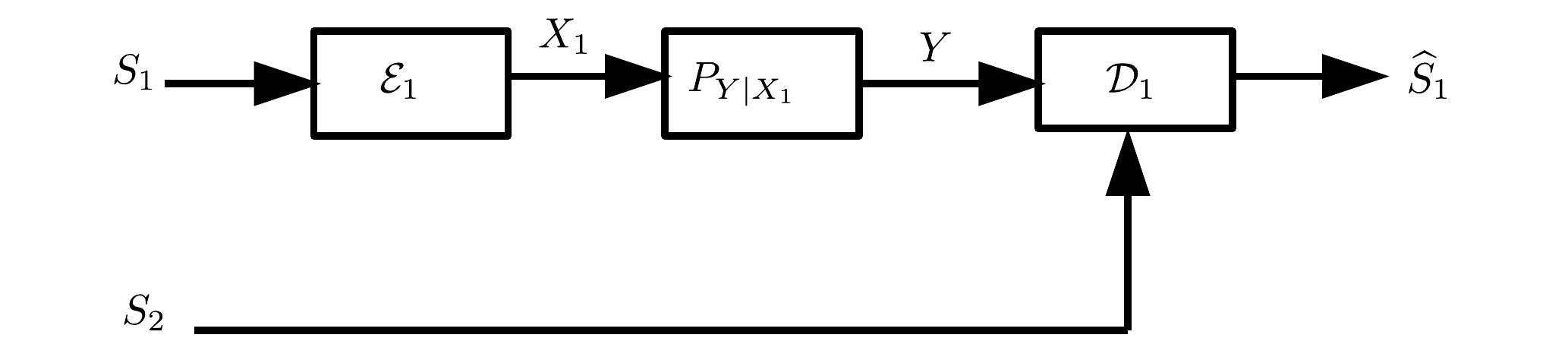} 
\caption{Lossless source coding with side information available only at the decoder}\label{Fig:sideinformation}
\end{figure}
\subsection{ Lossless Source Coding of $S_2$ with $S_1$ as the Side-Information}
Analogous to $\SID_{1|2}$, the finite blocklength lossless source coding problem of $S_2$ with $S_1$ as the side information can be posed as the following optimization problem,
$$\problemsmalla{$\SID_{2|1}$}
	{f,g}
	{\displaystyle \mathbb{E}[\Ibb\{S_2 \neq \Shat_2\}]}
				 {\begin{array}{r@{\ }c@{\ }l}
				 				 X_2&=&f_2(S_2)\\
								 \Shat_2&=&g(S_1,Y_2).
	\end{array}}
	$$
	The dual of the corresponding LP relaxation of $\SID_{2|1}$ then becomes,
	$$\maxproblemsmall{ $\DPSI_{2|1}$\hspace{0.17cm}}
	{ \gammat^a,\gammat^b,\lambdat^{(2|1)}_{\ssf},\lambdat^{(2|1)}_{\csf}}
	{\displaystyle \sum_{s_2}\gammat^a(s_2)+\sum_{y_2,s_1} \gammat^b(y_2,s_1)}
				 {\begin{array}{r@{\ }c@{\ }l}
				 				 				\gammat^a(s_2)-\sum_{y_2,s_1} \lambdat^{(2|1)}_{\csf}(s_1,s_2,x_2,y_2) &\leq & 0 \\ \forall x_2,s_2\quad  (C1) & &\\
				 				 				\gammat^b(s_1,y_2)-\sum_{s_2} \lambdat^{(2|1)}_{\ssf}(s_1,s_2,\shat_2,y_2) &\leq & 0\\		 				 				\forall s_1,\shat_2,y_2\quad (C2)& &\\				 				 				\lambdat^{(2|1)}_{\ssf}(s_1,s_2,\shat_2,y_2)+\lambdat^{(2|1)}_{\csf}(s_1,s_2,x_2,y_2)\hspace{-0.15cm}&\leq &\hspace{-0.15cm} \Delta(z') \\
				 				 				\forall z' \quad (C3),& &
\end{array}} $$ where $z':=(s_1,s_2,x_2,y_2,\shat_2)$, and $\Delta(z')=P_{S_1,S_2}(s_1,s_2)\Ibb\{y_2=x_2 \}\Ibb\{s_2\neq \shat_2\}.$
%
The following choice of source flow and channel flow results in a converse similar to the one in Theorem~\ref{thm:metaside12} for the problem $\SID_{2|1}$.
\begin{align}
\lambdat^{(2|1)}_{\csf}(s_1,s_2,x_2,y_2)&=\phi^{(2|1)}(s_1,s_2) \Ibb\{x_2=y_2\} \non\\
\lambdat^{(2|1)}_{\ssf}(s_1,s_2,\shat_2,y_2)&=-\phi^{(2|1)}(s_1,s_2)\Ibb\{s_2=\shat_2\},
\label{eq:constr4}
\end{align}where $\phi^{(2|1)}:\Sscr_1 \times \Sscr_2 \rightarrow [0,1]$ is such that $0\leq \phi^{(2|1)}(s_1,s_2) \leq P_{S_1,S_2}(s_1,s_2)$ for all $s_1\in \Sscr_1, s_2 \in \Sscr_2$.
\begin{theorem} \label{thm:metaside21} 
Consider the problem $\SID_{2|1}$. Consequently, 
for any code, the following lower bound holds,
\begin{align}
&\Ebb[\Ibb\{S_2\neq \Shat_2\}]\geq \OPT(\SID_{2|1})\geq  \OPT(\DPSI_{2|1})\non\\& \geq \sup_{0\leq \phi^{(2|1)}(s_1,s_2) \leq P_{S_1,S_2}(s_1,s_2)}\biggl \lbrace \sum_{s_1,s_2}\phi^{(2|1)}(s_1,s_2)\non \\& \qquad -M_2 \sum_{s_1}\max_{\shat_2}\phi^{(2|1)}(s_1,\shat_2) \biggr\rbrace. \label{eq:SID2converse}
\end{align}
\end{theorem}
In the next section, we extend the LP based framework to finite blocklength Slepian-Wolf coding problem and establish the duality based framework.
  \section{Linear Programming Based Framework for the Slepian-Wolf  Problem}\label{sec:LP}
In this section, we discuss the extension of the linear programming (LP) based framework in Section~\ref{sec:point-to-point} to the finite blocklength Slepian-Wolf coding problem SW. 
Towards this, consider the joint probability distribution $Q: \Sscr_1\times \Sscr_2 \times \Xscr_1\times \Xscr_2 \times \Yscr_1\times \Yscr_2 \times \Sscrhat_1 \times \Sscrhat_2 \rightarrow [0,1]$ which can be factored as,
 $$Q (z)\equiv P_{S_1,S_2}Q_{X_1|S_1}Q_{X_2|S_2}P_{Y_1Y_2|X_1X_2}Q_{\Shat_1,\Shat_2|Y_1,Y_2}(z),$$
 where $\Zscr:= \Sscr_1\times \Sscr_2 \times \Xscr_1\times \Xscr_2 \times \Yscr_1\times \Yscr_2 \times \Sscrhat_1 \times \Sscrhat_2$ and $z:=(s_1,s_2,x_1,x_2,y_1,y_2,\shat_1,\shat_2) \in \Zscr$.  Employing $Q$, we obtain the following optimization problem over joint probability distributions,
 $$\singlecolproblemsmalla{${\rm SW}$}
	{ \substack{Q_{X_1|S_1},Q_{X_2|S_2},\\Q,Q_{\Shat_1,\Shat_2|Y_1,Y_2}}}
	{\displaystyle \sum_{z \in \Zscr}\Ibb\{(s_1,s_2)\neq (\shat_1,\shat_2)\}Q(z)}
				 {\begin{array}{r@{\ }c@{\ }l}
				 P_{S_1S_2} Q_{X_1|S_1} Q_{X_2|S_2}P_{Y_1,Y_2|X_1,X_2}Q_{\Shat_1,\Shat_2|Y_1,Y_2}\hspace{-0.07cm}(z) \hspace{-0.085cm}&\equiv& \hspace{-0.085cm}Q(z),\\
Q_{X_1|S_1}\in\mathcal{P}(\Xscr_1|\Sscr_1),\\ Q_{X_2|S_2}\in \mathcal{P}(\Xscr_2|\Sscr_2), \\Q_{\Shat_1,\Shat_2|Y_1,Y_2}\in \mathcal{P}(\Sscrhat_1,\Sscrhat_2|\Yscr_1,\Yscr_2),
	\end{array}}$$
	 where, $\mathcal{P}(A|B)\hspace{-0.1cm}:=\hspace{-0.1cm}\{Q_{A|B}\hspace{-0.03cm}| \hspace{-0.03cm}\sum_{a}Q_{A|B}(a|b)=1, Q_{A|B}(a|b)\geq 0, \hspace{0.05cm}\forall\hspace{0.1cm} a,b\}.$	 
Here, $Q_{X_1|S_1}$ and $Q_{X_2|S_2}$ represent the two randomized encoders, and $Q_{\Shat_1,\Shat|Y_1,Y_2}$ represents a randomized decoder. It is easy to argue as in~\cite{jose2016linear} that the above formulation is in fact equivalent to problem SW stated in the introduction.


As in the case of the point-to-point problems, the presence of the multilinear constraint renders the feasible region of SW nonconvex. Notice that the degree of the multilinear term is three since there are three decision makers, whereas in the point-to-point problems the degree was two. To obtain converses or lower bounds on the optimal value of SW, we will again derive a linear programming (LP) relaxation of the nonconvex feasible region of SW, as shown in the next section. 
\subsection{LP Relaxation}
For obtaining a linear programming relaxation of SW, we  
resort to the ``lift-and-project'' technique in integer programming.  Towards this, we define the following new variables, \begin{align}
	W(z) &\equiv Q_{X_1|S_1} Q_{X_2|S_2} Q_{\Shat_1,\Shat_2|Y_1,Y_2}(z),\label{eq:W}\\
U(z_1)&\equiv Q_{X_1|S_1}(x_1|s_1)Q_{\Shat_1,\Shat_2|Y_1,Y_2}(\shat_1,\shat_2|y_1,y_2),\label{eq:U}\\
	V(z_2)&\equiv Q_{X_2|S_2}(x_2|s_2)Q_{\Shat_1,\Shat_2|Y_1,Y_2}(\shat_1,\shat_2|y_1,y_2),\label{eq:V}\\
	T(z_3)&\equiv Q_{X_1|S_1}(x_1|s_1)Q_{X_2|S_2}(x_2|s_2),\label{eq:T}
\end{align} where recall that $z:=(s_1,s_2,x_1,x_2,y_1,y_2,\shat_1,\shat_2)$, $z_1:=(s_1,x_1,y_1,y_2,\shat_1,\shat_2)$, $z_2:=(s_2,x_2,\shat_1,\shat_2,y_1,y_2)$ and $z_3:=(s_1,s_2,x_1,x_2)$. Using these variables, we first lift the problem SW to a higher dimensional space and then impose additional valid constraints involving these new variables. 
To obtain these constraints, we adopt the following procedure.
 
 \begin{figure*}[!t]
	$$\problemsmallLP{LPSW}
	{\substack{W,T,V,U,\\Q_{X_1|S_1},Q_{X_2|S_2},\\Q_{\Shat_1,\Shat_2|Y_1,Y_2}}}
	{\displaystyle \sum_{z \in \Zscr}\Ibb\{(s_1,s_2)\neq (\shat_1,\shat_2)\}$ $P_{S_1,S_2}(s_1,s_2)$ $\Ibb\{(y_1,y_2)=(x_1,x_2)\} W(z)}
				 {\begin{array}{r@{\ }c@{\ }l}
				 				 				 \sum_{x_1} Q_{X_1|S_1}(x_1|s_1)&=& 1 :\hspace{0.1cm} \gamma^a(s_1) \hspace{3.42cm} \forall s_1 \\
\sum_{x_2}Q_{X_2|S_2}(x_2|s_2)&=&1 : \hspace{0.1cm}\gamma^b(s_2) \hspace{3.45cm}\forall s_2 \\
\sum_{\shat_1,\shat_2}Q_{\Shat_1,\Shat_2|Y_1,Y_2}(\shat_1,\shat_2|y_1,y_2)&=&1 :\hspace{0.1cm} \gamma^c(y_1,y_2) \hspace{3cm}\forall y_1,y_2 \\
\sum_{x_1}T(z_3)-Q_{X_2|S_2}(x_2|s_2) &= & 0: \hspace{0.1cm}\mu^{(2|1)}_{\csf}(x_2,s_1,s_2) \hspace{2.01cm} \forall x_2,s_1,s_2 \\
\sum_{x_1}U(z_1)-Q_{\Shat_1,\Shat_2|Y_1,Y_2}(\shat_1,\shat_2|y_1,y_2)&=&0 :\hspace{0.1cm} \mu^{(1)}_{\ssf}(s_1,\shat_1,\shat_2,y_1,y_2)\hspace{1.34cm}\forall s_1,\shat_1,\shat_2,y_1,y_2\\
\sum_{x_1} W(z)-V(z_2)&=& 0: \hspace{0.1cm}\lambda^{(2|1)}_{\ssf}(s_1,s_2,x_2,y_1,y_2,\shat_1,\shat_2) \hspace{0.1cm} \forall s_1,s_2,x_2,y_1,y_2,\shat_1,\shat_2\\ 
\sum_{x_2}T(z_3)-Q_{X_1|S_1}(x_1|s_1) &=&0 :\hspace{0.1cm} \mu^{(1|2)}_{\csf}(x_1,s_1,s_2) \hspace{2.01cm} \forall x_1,s_1,s_2  \\
\sum_{x_2} V(z_2)-Q_{\Shat_1,\Shat_2|Y_1,Y_2}(\shat_1,\shat_2|y_1,y_2) &=&0 :\hspace{0.1cm} \mu^{(2)}_{\ssf}(s_2,\shat_1,\shat_2,y_1,y_2)\hspace{1.3cm} \forall s_1,\shat_1,\shat_2,y_1,y_2 \\ 
\sum_{x_2}W(z)-U(z_1)&=& 0 : \hspace{0.1cm}\lambda^{(1|2)}_{\ssf}(s_1,s_2,x_1,y_1,y_2,\shat_1,\shat_2) \hspace{0.1cm} \forall s_1,s_2,x_1,y_1,y_2,\shat_1,\shat_2 \\
\sum_{\shat_2,\shat_1}W(z)-T(z_3) &=&0 :\hspace{0.1cm} \lambda_{\csf}(s_1,s_2,x_1,x_2,y_1,y_2) \hspace{0.99cm} \forall s_1,s_2,x_1,x_2,y_1,y_2\\
\sum_{\shat_2,\shat_1}V(z_2)-Q_{X_2|S_2}(x_2|s_2) &=&0: \hspace{0.1cm}\mu^{(2)}_{\csf}(x_2,s_2,y_1,y_2) \hspace{1.75cm} \forall  x_2,s_2,y_1,y_2\\
\sum_{\shat_1,\shat_2} U(z_1)-Q_{X_1|S_1}(x_1|s_1) &=&0:\hspace{0.1cm} \mu^{(1)}_{\csf}(x_1,s_1,y_1,y_2) \hspace{1.755cm} \forall x_1,s_1,y_1,y_2\\
 Q_{X_1|S_1},Q_{X_2|S_2},Q_{\Shat_1,\Shat_2|Y_1,Y_2}, V,U,W,T &\geq & 0.
	\end{array}}$$ 
	\end{figure*}
 For each $s_1 \in \Sscr_1$, we multiply both sides of the constraint  $\sum_{x_1} Q_{X_1|S_1}(x_1|s_1)= 1$ by  $Q_{X_2|S_2}(x_2|s_2)$ for all $x_2,s_2$, by $Q_{\Shat_1,\Shat_2|Y_1,Y_2}(\shat_1,\shat_2|y_1,y_2)$ for all $\shat_1,\shat_2,y_1,y_2$ and by $Q_{X_2|S_2}(x_2|s_2)Q_{\Shat_1,\Shat_2|Y_1,Y_2}(\shat_1,\shat_2|y_1,y_2)$ for all $x_2,s_2,\shat_1,\shat_2,y_1,y_2$, to obtain three new sets of multilinear equality constraints. Replace the resulting multilinear product terms by the newly defined variables in \eqref{eq:W}--\eqref{eq:T}. This results in the following new valid linear constraints in the lifted space,
%
\begin{align*}
\sum_{x_1}T(x_1,x_2,s_1,s_2) &\equiv Q_{X_2|S_2}(x_2|s_2)\\
\sum_{x_1}U(x_1,s_1,\shat_1,\shat_2,y_1,y_2)&\equiv Q_{\Shat_1,\Shat_2|Y_1,Y_2}(\shat_1,\shat_2|y_1,y_2)\\
\sum_{x_1} W(z)&\equiv V(x_2,s_2,\shat_1,\shat_2,y_1,y_2).
\end{align*} 
Similar set of linear constraints can be obtained corresponding to $\sum_{x_2}Q_{X_2|S_2}(x_2|s_2)=1$ for all $s_2$ and $\sum_{\shat_1,\shat_2}Q_{\Shat_1,\Shat_2|Y_1,Y_2}(\shat_1,\shat_2|y_1,y_2)=1$ for all $ y_1,y_2$.
 Subsequently, add these new sets of linear constraints to the original constraints of SW. Further, replace $Q$ in the objective function of SW with the first constraint written in terms of $W$, drop the multilinear equalities in \eqref{eq:W}-\eqref{eq:T} and we have the LP relaxation of SW, LPSW as given in the next page. Notice that  the constraints of LPSW are implied by the constraints of SW whereby LPSW is a relaxation of SW.
 
Here, $\gamma^a(s_1)$, $\gamma^b(s_2)$, $\gamma^c(y_1,y_2)$, $\mu^{(2|1)}_{\csf}(x_2,s_1,s_2)$, $\mu^{(1)}_{\ssf}(s_1,\shat_1,\shat_2,y_1,y_2)$, $\lambda^{(2|1)}_{\ssf}(s_1,s_2,x_2,y_1,y_2,\shat_1,\shat_2)$, $\mu^{(1|2)}_{\csf}(x_1,s_1,s_2)$, $\mu^{(2)}_{\ssf}(s_2,\shat_1,\shat_2,y_1,y_2)$, $\lambda^{(1|2)}_{\ssf}(s_1,s_2,x_1,y_1,y_2,\shat_1,\shat_2)$,  $\lambda_{\csf}(s_1,s_2,x_1,x_2,y_1,y_2)$, $\mu^{(2)}_{\csf}(s_2,x_2,y_1,y_2)$ and $\mu^{(1)}_{\csf}(s_1,x_1,y_1,y_2)$ represent the Lagrange multipliers corresponding to the constraints of LPSW in that order.
\subsubsection{An stronger optimal transport interpretation}\label{sec:opttransSW} 
The LP relaxation of SW, LPSW, also admits an interpretation via a ``multiterminal'' optimal transport problem. 
As in the point-to-point LP relaxation, we note that for each $s_1 \in \Sscr_1,s_2\in \Sscr_2,y_1 \in \Yscr_1, y_2 \in \Yscr_2$, $W(s_1,s_2,\cdot,\cdot,y_1,y_2,\cdot,\cdot)$ is  a coupling on $\Xscr_1 \times \Xscr_2 \times (\Sscrhat_1 \times \Sscrhat_2)$ between all three marginals $Q_{X_1|S_1}(\cdot|s_1)$, $Q_{X_2|S_2}(\cdot|s_2)$ and $Q_{\Shat_1,\Shat_2|Y_1,Y_2}(\cdot,\cdot|y_1,y_2)$; denote the set of all such $W$ by $\Xi'$. However, note that $W \in \Xi'$ does not automatically imply that, for instance, 
\begin{equation}
\sum_{x_1}W(z) \in \Xi(Q_{X_2|S_2},Q_{\Shat_1,\Shat_2|Y_1,Y_2}), \label{eq:opttrans1} 
\end{equation} and likewise, that 
\begin{align}
\sum_{x_2}W(z)&\in \Xi(Q_{X_1|S_1},Q_{\Shat_1,\Shat_2|Y_1,Y_2}), \label{eq:opttrans2} \\ 
\sum_{\shat_1,\shat_2}W(z) &\in \Xi(Q_{X_1|S_1},Q_{X_2|S_2}). \label{eq:opttrans3} 
\end{align}
 LPSW is obtained by imposing not only that $W \in \Xi'$, but also \eqref{eq:opttrans1}-\eqref{eq:opttrans3}. Skipping the latter requirements would evidently lead to a looser relaxation which would perhaps not suffice for our purpose of obtaining tight converses. As in the point-to-point problem, LPSW is a nested minimization where the relaxation arises from replacing the product of kernels of randomized codes by any coupling in $\Xi'$ that is constrained by \eqref{eq:opttrans1}-\eqref{eq:opttrans3}, and then minimizing over all codes.

We note that the variables $U,V,T$ in \eqref{eq:U}-\eqref{eq:T} are introduced in LPSW only to express the constraints \eqref{eq:opttrans1}-\eqref{eq:opttrans3} on $W$ in a clearer manner. They could be eliminated in a straightforward manner and the entire problem could be expressed only in terms of $W$ and the randomized code  $(Q_{X_1|S_1},Q_{X_2|S_2},Q_{\Shat_1,\Shat_2|Y_1,Y_2}).$

%

\subsection{Duality and Converses}
Employing the Lagrange multipliers corresponding to the constraints of LPSW, we now obtain the dual of LPSW, denoted DPSW, and shown on the next page. Here, $\Pi(z)\equiv$ $ \Ibb\{(s_1,s_2)\neq (\shat_1,\shat_2)\}\times$ $P_{S_1,S_2}(s_1,s_2)$ $\Ibb\{(y_1,y_2)=(x_1,x_2)\}.$
\begin{figure*}[!t]
\begin{small}
	 $$\maxproblemlarge{DPSW}
	{ \Theta}
	{\displaystyle \sum_{s_1}\gamma^a(s_1)+\sum_{s_2}\gamma^b(s_2)+\sum_{y_1,y_2}\gamma^c(y_1,y_2)}
				 {\hspace{-1cm}\begin{array}{r@{\ }c@{\ }l}
				 				 				\gamma^a(s_1)-\sum_{y_1,y_2} \mu^{(1)}_{\csf}(s_1,x_1,y_1,y_2)-\sum_{s_2}\mu^{(1|2)}_{\csf}(x_1,s_1,s_2) \leq  0 \quad \quad \forall x_1,s_1 \hspace{1.17cm}\quad \hspace{0.3cm}  (D1)\\
				 				 				\gamma^b(s_2)-\sum_{y_1,y_2} \mu^{(2)}_{\csf}(s_2,x_2,y_1,y_2)-\sum_{s_1}\mu^{(2|1)}_{\csf}(x_2,s_1,s_2) \leq  0\quad \quad \forall x_2,s_2\hspace{1.17cm}\hspace{0.3cm}  \quad(D2)\\
				 				 				\gamma^c(y_1,y_2)-\sum_{s_2}\mu^{(2)}_{\ssf}(s_2,\shat_1,\shat_2,y_1,y_2)-\sum_{s_1}\mu^{(1)}_{\ssf}(s_1,\shat_1,\shat_2,y_1,y_2) \leq  0\quad \quad\forall \shat_1,\shat_2,y_1,y_2 \hspace{0.6cm}\hspace{0.3cm} (D3)\\
				 				 				\lambda^{(1|2)}_{\ssf}(s_1,s_2,x_2,y_1,y_2,\shat_1,\shat_2)+\lambda^{(2|1)}_{\ssf}(s_1,s_2,x_1,y_1,y_2,\shat_1,\shat_2)+\lambda_{\csf}(s_1,s_2,x_1,x_2,y_1,y_2) \leq  \Pi(z) \hspace{0.2cm} \forall z\hspace{1.27cm}\quad \quad\hspace{0.45cm}  (D4)\\
				 				 				\mu^{(2)}_{\ssf}(s_2,\shat_1,\shat_2,y_1,y_2)+\mu^{(2)}_{\csf}(s_2,x_2,y_1,y_2)-\sum_{s_1}\lambda^{(1|2)}_{\ssf}(s_1,s_2,x_2,y_1,y_2,\shat_1,\shat_2) \leq  0 \quad \quad \forall s_2,x_2,y_1,y_2,\shat_1,\shat_2\hspace{0.03cm}(D5)\\
	\mu^{(1)}_{\ssf}(s_1,\shat_1,\shat_2,y_1,y_2)+\mu^{(1)}_{\csf}(s_1,x_1,y_1,y_2)-\sum_{s_2}		 \lambda^{(2|1)}_{\ssf}(s_1,s_2,x_1,y_1,y_2,\shat_1,\shat_2)\leq  0 \quad \quad \forall s_1,x_1,y_1,y_2,\shat_1,\shat_2\hspace{0.03cm}(D6)\\
	\mu^{(2|1)}_{\csf}(x_2,s_1,s_2)+\mu^{(1|2)}_{\csf}(x_1,s_1,s_2)-\sum_{y_1,y_2}	\lambda_{\csf}(s_1,s_2,x_1,x_2,y_1,y_2)\leq 0 \quad \quad \forall x_1,x_2,s_1,s_2	\quad\hspace{0.58cm} (D7)		
\end{array}}$$ \end{small}\vspace{-0.7cm}\end{figure*} 
Let $\Theta:=(\lambda^{(1|2)}_{\ssf},\lambda^{(2|1)}_{\ssf},\lambda_{\csf}, \gamma^a,\gamma^b,\gamma^c,\mu^{(1)}_{\ssf},\mu^{(1)}_{\csf},\mu^{(2)}_{\ssf},\mu^{(2)}_{\csf},\mu^{(1|2)}_{\csf},\mu^{(2|1)}_{\csf})$ represent the collection of all these Lagrange multipliers or dual variables.

To evaluate the optimal value of DPSW, it suffices to take $\gamma^a(s_1)$, $\gamma^b(s_2)$ and $\gamma^c(y_1,y_2)$ such that the constraints (D1), (D2), (D3) hold with equality. Thus, at optimality,
\begin{small}
\begin{align*}
\gamma^a(s_1) &\equiv \min_{x_1}\biggl \lbrace\sum_{y_1,y_2} \mu^{(1)}_{\csf}(s_1,x_1,y_1,y_2)+\sum_{s_2}\mu^{(1|2)}_{\csf}(x_1,s_1,s_2)\biggr \rbrace \\
\gamma^b(s_2) & \equiv \min_{x_2} \biggl \lbrace \sum_{y_1,y_2} \mu^{(2)}_{\csf}(s_2,x_2,y_1,y_2) +\sum_{s_1}\mu^{(2|1)}_{\csf}(x_2,s_1,s_2)\biggr \rbrace  \\
\gamma^c(y_1,y_2) &\equiv \min_{\shat_1,\shat_2}\biggl \lbrace \sum_{s_2}\mu^{(2)}_{\ssf}(s_2,\shat_1,\shat_2,y_1,y_2) \hspace{-0.1cm} +\hspace{-0.1cm} \sum_{s_1}\mu^{(1)}_{\ssf}(s_1,\shat_1,\shat_2,y)\biggr \rbrace.
\end{align*}
\end{small}
Let $\bar{\Theta}:=(\lambda^{(1|2)}_{\ssf},\lambda^{(2|1)}_{\ssf}, \lambda_{\csf} ,\mu^{(1)}_{\ssf},\mu^{(1)}_{\csf},  \mu^{(2)}_{\ssf},\mu^{(2)}_{\csf},\mu^{(1|2)}_{\csf},\mu^{(2|1)}_{\csf})$ represent the collection of remaining dual variables. 
From the duality of linear programming, the following lemma then outlines our framework for obtaining lower bounds.
\begin{lemma}\label{lemma:dualbound}
Any collection of functions $\bar{\Theta}$ 
 satisfying constraints (D4)-(D7) yields the following lower bound on the optimal value of SW, 
\ie,
\begin{align}&\OPT({\rm SW})\stackrel{(a)}{\geq} \OPT({\rm LPSW})\stackrel{(b)}{=}\OPT({\rm DPSW})\non\\&\geq 
\sum_{s_1}\min_{x_1}\biggl \lbrace\sum_{y_1,y_2} \mu^{(1)}_{\csf}(s_1,x_1,y_1,y_2)+\sum_{s_2}\mu^{(1|2)}_{\csf}(x_1,s_1,s_2)\biggr \rbrace \non \non\\& + \sum_{s_2} \min_{x_2} \biggl \lbrace \sum_{y_1,y_2} \mu^{(2)}_{\csf}(s_2,x_2,y_1,y_2) +\sum_{s_1}\mu^{(2|1)}_{\csf}(x_2,s_1,s_2)\biggr \rbrace  \non \\&+\sum_{y_1,y_2} \min_{\shat_1,\shat_2}\biggl \lbrace \hspace{-0.1cm}\sum_{s_2}\mu^{(2)}_{\ssf}(s_2,\shat_1,\shat_2,y_1,y_2) \non\\&+\sum_{s_1}\mu^{(1)}_{\ssf}(s_1,\shat_1,\shat_2,y)\biggr \rbrace. \label{eq:SWframework}
	\end{align}
	\end{lemma}
	The inequality in (a) follows since LPSW is a relaxation of SW and (b) results from the duality of linear programming. The last inequality follows from $\bar{\Theta}$ being feasible for DPSW and using that (D1)-(D3) hold with equality.
		
Thus, to obtain finite blocklength lower bounds on SW, it suffices to construct functions,
\begin{align} \lambda^{(2|1)}_{\ssf} &: \Sscr_1\times \Sscr_2 \times \Xscr_2\times \Yscr_1\times \Yscr_2 \times \Sscrhat_1\times \Sscrhat_2 \rightarrow \Real, \non 
\\ 
\lambda^{(1|2)}_{\ssf} &: \Sscr_1\times \Sscr_2 \times \Xscr_1\times \Yscr_1\times \Yscr_2 \times \Sscrhat_1\times \Sscrhat_2 \rightarrow \Real,\non\\ 
 	\lambda_{\csf} &: \Sscr_1\times \Sscr_2 \times \Xscr_1\times \Xscr_2 \times \Yscr_1\times \Yscr_2  \rightarrow \Real, \non
 	\\
   \mu^{(1)}_{\ssf} &:\Sscr_1 \times \Sscrhat_1\times \Sscrhat_2 \times \Yscr_1\times \Yscr_2 \rightarrow \Real,\non \\
   \mu^{(1)}_{\csf} &:\Sscr_1 \times \Xscr_1 \times \Yscr_1\times \Yscr_2 \rightarrow \Real, \label{eq:dpflows} \\
 		\mu^{(2)}_{\ssf}& :\Sscr_2 \times \Sscrhat_1\times \Sscrhat_2 \times \Yscr_1\times \Yscr_2 \rightarrow \Real,\non\\
   \mu^{(2)}_{\csf} &:\Sscr_2 \times \Xscr_2 \times \Yscr_1\times \Yscr_2 \rightarrow \Real, \non \\
 	  \mu^{(2|1)}_{\csf}&:\Xscr_2 \times \Sscr_1 \times \Sscr_2 \rightarrow \Real,    \non \\
\mu^{(1|2)}_{\csf}&:\Xscr_1 \times \Sscr_1 \times \Sscr_2 \rightarrow \Real, \non 
   \end{align}
such that the point-wise inequalities in (D4)-(D7) are satisfied. 

We call the above collection of functions $\bar{\Theta}$, a feasible point of DPSW.  As is evident, construction of such a feasible point of DPSW is challenging and probably cumbersome at first glance. Another hindrance is the difficulty in interpreting these variables so as to develop any intuitions on construction of these variables.

Consequently, in this paper, we present a systematic method to construct feasible points of DPSW and thereby, obtain finite blocklength converses for SW coding.  
 We show that a combination of the source and channel flows of the problems $\DPJE$, $\DPSI_{1|2}$ and $\DPSI_{2|1}$, yields a new feasible point of DPSW and thereby, a new finite blocklength converse. We discuss this in the next section.
 
\section{From Point-to-Point Converses to Slepian-Wolf Converses}\label{sec:sythesize}
In this section, we present a systematic synthesis of finite blocklength converses for the Slepian-Wolf coding problem from the source and channel flows of the dual programs $\DPJE$, $\DPSI_{1|2}$ and $\DPSI_{2|1}$ discussed in Section~\ref{sec:point-to-point}.

We begin by discussing the structure of DPSW. Since constraints (D1)-(D3) can be assumed to hold with equality, our main concern is with the variables $\lambda_\ssf^{(1|2)},\lambda_\ssf^{(2|1)},\lambda_\csf$ and $\mu_\ssf^{(1)},\mu_\ssf^{(2)},\mu_\csf^{(1)},\mu_\csf^{(2)}, \mu^{(2|1)}_\csf,\mu_\csf^{(1|2)}.$ We will refer to these variables (recall that these are functions, as stated in \eqref{eq:dpflows}) also as \textit{flows}. Our approach for interpreting and classifying these flows is based on relating these flows to flows of problems $\DPJE$, $\DPSI_{1|2}$ and $\DPSI_{2|1}$. We remark that there may be other approaches that would yield a more refined understanding.

We begin with the $\lambda$'s. Consider the flow $\lambda^{(1|2)}_{\ssf}$. 
Observe that  $\lambda^{(1|2)}_{\ssf}$ is a function of $s_1,s_2,x_2,y_1,y_2,\shat_1,\shat_2$ but is independent of $x_1$. Hence, for each fixed $s_2,x_2,y_2$ and $\shat_2$, $\lambda^{(1|2)}_{\ssf}$ may be likened to a source flow from $S_1$ to $\Shat_1$ (recall that the source flow in the point-to-point problem was a function of the source, the channel output and the destination, but \textit{not} of the channel input). The dependence of $\lambda^{(1|2)}_{\ssf}$ on $s_2,x_2,y_2,\shat_2$ hints at (coded or uncoded) side-information about $S_2$ through the path $S_2 \rightarrow X_2 \rightarrow Y_2 \rightarrow \Shat_2$. This leads one to surmise that the flow $\lambda^{(1|2)}_\ssf$ would have a close relation to the source flow of the problem $\DPSI_{1|2}.$ Thus we refer to $\lambda^{(1|2)}_\ssf$ as a source flow for DPSW.
Note though, that this is not the only heuristic one can apply. If $\lambda^{(1|2)}_{\ssf} $ is assumed to be also independent of $x_2$, then $\lambda^{(1|2)}_{\ssf}(s_1,s_2,x_2,y_1,y_2,\shat_1,\shat_2)$ can be also interpreted to be the source flow in the problem $\DPJE$ where $S_1,S_2$ are jointly encoded. Thus one surmises that a value for $\lambda_\ssf^{(1|2)}$ could probably be arrived at by a combination of the source flows of $\DPJE$ and $\DPSI_{1|2}.$
A similar heuristic can be applied to surmise that $\lambda^{(2|1)}_{\ssf}$ could be arrived at through the source flows of $\DPJE$ and $\DPSI_{2|1}$. The flow $\lambda_{\csf}(s_1,s_2,x_1,x_2,y_1,y_2)$ which depends on correlated sources and the channel inputs and outputs, appears to be related to the channel flows of all three problems $\DPJE,\DPSI_{1|2}$ and $\DPSI_{2|1}$, and should therefore be a function of the latter flows. We refer to it as the channel flow. Thus in problem DPSW, there are \textit{two} source flows and one channel flow that satisfy an error density bottleneck (D4).

We now come to the $\mu$'s. Notice that the flows of problem DPSW fall into a hierarchy wherein the $\lambda$'s are constrained by the error density bottleneck (constraint (D4)), whereas the $\mu$'s are constrained by a bottleneck determined by the $\lambda$'s. Arguing as in the case of the $\lambda$'s we see that $\mu_\ssf^{(1)}$ and $\mu^{(1)}_\csf$ are akin to source and channel flows of a coding problem along the path $S_1\rightarrow X_1 \rightarrow Y_1 \rightarrow \Shat_1$. Note though the \textit{objective} of the problem (source coding, or something else) would depend on $\lambda_\ssf^{(2|1)}$, the RHS of constraint (D6). Likewise $\mu_\ssf^{(2)}$ and $\mu^{(2)}_\csf$ resemble source channel flows for a coding problem along $S_2 \rightarrow X_2 \rightarrow Y_2 \rightarrow \Shat_2$ whose objective is determined by $\lambda_\ssf^{(1|2)}$. The final set of dual variables $\mu^{(1|2)}_\csf$ and $\mu^{(2|1)}_\csf$ are somewhat distinct from the rest, since they do not seem to be analogous to any flows from point-to-point problems. We will interpret these later.

The following two propositions distill these heuristics into a formal relationship between the feasible regions of problems DPSW and problems $\DPJE$, $\DPSI_{1|2}$ and $\DPSI_{2|1}$.

 \begin{proposition}\label{prop:sideconverse}
Let $\bar{\Theta}_{1|2}:=(\bar{\gamma}^a,\bar{\gamma}^b,\bar{\lambda}^{(1|2)}_{\ssf},\bar{\lambda}^{(1|2)}_{\csf})\in \FEA(\DPSI_{1|2})$ with its corresponding objective value, $\OBJ(\DPSI_{1|2})=\sum_{s_1}\gammab^a(s_1)+\sum_{s_2,y_1}\gammab^b(s_2,y_1)$. Then the following choice of values for the variables of DPSW is feasible.
\begin{align}\hspace{-0.1cm} \lambda^{(1|2)}_{\ssf}\hspace{-0.05cm}(s_1,s_2,x_2,y_1,y_2,\shat_1,\shat_2)\hspace{-0.05cm}&\equiv \hspace{-0.05cm}\bar{\lambda}^{(1|2)}_{\ssf}\hspace{-0.05cm}(s_1,s_2,\shat_1,y_1)\Ibb\{x_2=y_2\}\non\\
 \lambda_{\csf} (s_1,s_2,x_1,x_2,y_1,y_2)&\equiv\hspace{-0.05cm} \bar{\lambda}^{(1|2)}_{\csf}\hspace{-0.05cm}(s_1,s_2,x_1,y_1)\hspace{-0.05cm}\Ibb\{x_2=y_2\} , \non\\ \mu^{(2)}_{\csf}(s_2,x_2,y_1,y_2)&\equiv \bar{\gamma}^b(s_2,y_1)\Ibb\{x_2=y_2\},\hspace{-0.3cm}\label{eq:constr1}\\  \mu^{(1|2)}_{\csf}(x_1,s_1,s_2)&\equiv \sum_{y_1}\bar{\lambda}^{(1|2)}_{\csf}(s_1,s_2,x_1,y_1)\non,\\
\hspace{-0.1cm}\lambda^{(2|1)}_{\ssf}(s_1,s_2,x_1,y_1,y_2,\shat_1,\shat_2) &\equiv 0,\quad \mu^{(2|1)}_{\csf} (x_2,s_1,s_2)\equiv 0 \non, \\
\mu^{(1)}_{\ssf}(s_1,\shat_1,\shat_2,y_1,y_2) &\equiv 0, \quad \mu^{(1)}_{\csf}(s_1,x_1,y_1,y_2)\equiv 0,\non \\
\mu^{(2)}_{\ssf}(s_2,\shat_1,\shat_2,y_1,y_2)& \equiv 0,\quad \gamma^c(y_1,y_2)\equiv 0, \non
\end{align}
\begin{align}
\gamma^b(s_2)&\equiv \sum_{y_1} \bar{\gamma}^b(s_2,y_1)\non, \qquad
\gamma^a(s_1) \equiv \bar{\gamma}^a(s_1). \hspace{0.02cm}\non
\end{align}
Consequently, $\OPT(\DPSW)\geq \sum_{s_1}\gamma^a(s_1)+\sum_{s_2}\gamma^b(s_2)+\sum_{y_1,y_2}\gamma^c(y_1,y_2) $ which is equal to the objective of $\DPSI_{1|2}$ under $\bar{\Theta}_{1|2}$. In particular, considering $\bar{\Theta}_{1|2}$ as the optimal solution of $\DPSI_{1|2}$ gives $\OPT(\DPSW)\geq \OPT(\DPSI_{1|2})$.
\end{proposition}
\begin{proof}
The proof is included in Appendix~\ref{sec:appB}.
\end{proof}
It thus becomes clear that given any feasible point $(\bar{\gamma}^a,\bar{\gamma}^b,\bar{\lambda}^{(1|2)}_{\ssf},\bar{\lambda}^{(1|2)}_{\csf})\in \FEA(\DPSI_{1|2})$, one can construct a feasible point of DPSW as given in \eqref{eq:constr1}. Moreover, the objective value of the resulting feasible point gives a lower bound on $\OPT(\DPSW)$. Similarly, it can be shown that given any feasible point of $\DPSI_{2|1}$, one can construct a feasible point of $\DPSW$ with the same cost, thereby implying $$ \OPT(\DPSW) \geq \OPT(\DPSI_{2|1}).$$

As with the problems with side-information, the following proposition illustrates a relation between the feasible regions of DPJE and DPSW.
\begin{proposition}\label{prop:JEnDPSW}
The following relationship between the feasible region of DPJE and DPSW holds.
Let $\hat{\Theta}:=(\hat{\gamma}^a,\hat{\gamma}^b,\hat{\lambda}_{\ssf},\hat{\lambda}_{\csf})\in \FEA(\DPJE)$ with its corresponding objective value, $\OBJ(\DPJE)=\sum_{s_1,s_2}\gammah^a(s_1,s_2)+\sum_{y_1,y_2}\gammah^b(y_1,y_2)$. Then the following choice of values of the variables of DPSW are feasible.
\begin{align*}
\hspace{-0.1cm} \lambda^{(1|2)}_{\ssf}\hspace{-0.05cm}(s_1,s_2,x_2,y_1,y_2,\shat_1,\shat_2)\hspace{-0.05cm}&\equiv \hspace{-0.05cm}\lambdah_{\ssf}(s_1,s_2,\shat_1,\shat_2,y_1,y_2)\\ \lambda_{\csf} (s_1,s_2,x_1,x_2,y_1,y_2)&\equiv \lambdah_{\csf}(s_1,s_2,x_1,x_2,y_1,y_2), \non\\ \mu^{(2)}_{\csf}(s_2,x_2,y_1,y_2)&\equiv 0, \quad  \mu^{(1|2)}_{\csf}(x_1,s_1,s_2)\equiv 0 \non,\\
\hspace{-0.1cm}\lambda^{(2|1)}_{\ssf}(s_1,s_2,x_1,y_1,y_2,\shat_1,\shat_2) &\equiv 0,\\\mu^{(2|1)}_{\csf} (x_2,s_1,s_2)&\equiv \gammah^a(s_1,s_2) \non,\\
\mu^{(1)}_{\ssf}(s_1,\shat_1,\shat_2,y_1,y_2) &\equiv 0, \quad \mu^{(1)}_{\csf}(s_1,x_1,y_1,y_2)\equiv 0,\non
\end{align*}
\begin{align*}
\gamma^a(s_1)&\equiv 0, \quad\gamma^b(s_2)\equiv \sum_{s_1} \gammah^a(s_1,s_2)\non,\\
\gamma^c(y_1,y_2)&\equiv \gammah^b(y_1,y_2),\\   \mu^{(2)}_{\ssf}(s_2,\shat_1,\shat_2,y_1,y_2) &\equiv \sum_{s_1}\lambdah_{\ssf}(s_1,s_2,\shat_1,\shat_2,y_1,y_2).
\end{align*}
 Consequently, $\OPT(\DPSW)\geq \sum_{s_1}\gamma^a(s_1)+\sum_{s_2}\gamma^b(s_2)+\sum_{y_1,y_2}\gamma^c(y_1,y_2)$ which is the objective of $\DPJE$ under $\hat{\Theta}.$  Moreover, considering $\hat{\Theta}$ to be the optimizing feasible point of DPJE yields that $\OPT(\DPSW)\geq \OPT(\DPJE)$.
\end{proposition}
\begin{proof}
The proof is similar to the proof of Proposition~\ref{prop:sideconverse} and we skip the proof here.
\end{proof}

The relationships between the feasible regions of DPSW with that of DPJE, $\DPSI_{1|2}$ and $\DPSI_{2|1}$ established through Propositions~\ref{prop:sideconverse} and \ref{prop:JEnDPSW} help in establishing a formal interpretation for the roles of the dual variables of DPSW. From Proposition~\ref{prop:sideconverse}, we see that the dual variable $\lambda^{(1|2)}_{\ssf}$ that we considered as akin to a source flow for $\DPSI_{1|2}$  and also as a source flow for the $\DPJE$, has a somewhat more complex interpretation. Specifically, while the latter interpretation  holds thanks to Proposition~\ref{prop:JEnDPSW}, the former holds only along the diagonal $x_2=y_2$, as seen in \eqref{eq:constr1}. A similar caveat holds for the channel flow $\lambda_{\csf}$ and the other source flow $\lambda^{(2|1)}_{\ssf}$.


We remark here that the choices of the $\mu$'s in Propositions~\ref{prop:JEnDPSW} and \ref{prop:sideconverse} are not necessarily optimal and hence, one may not obtain a sharp interpretation for these flows. Nonetheless, it can be seen that when the sum $\sum_{s_1}\lambda^{(1|2)}_{\ssf}$ is considered as in constraint (D5), the source flow from $S_1$ to destination $\Shat_1$ is averaged out, and what is left, is the influence of side-information of $S_2$. From Proposition~\ref{prop:JEnDPSW}  it can be seen that $\mu^{(2)}_{\ssf}$ accounts for the point-to-point like source flow through the path from $S_2$ to the destination node $(\Shat_1,\Shat_2)$. Further, \eqref{eq:constr1} implies that $\mu^{(2)}_{\csf}$ accounts for the channel flow through this path. 
 Similarly,  $\mu^{(1)}_{\ssf}$ represents the point-to-point like {source flow} through the path from $S_1$ to $(\Shat_1,\Shat_2)$ via the decoder and $\mu^{(1)}_{\csf}$ represents the corresponding {channel flow} for this path.
Finally, thanks to the relation $\mu^{(1|2)}_{\csf} = \sum_{y_1} \bar{\lambda}_\csf^{(1|2)}$ in \eqref{eq:constr1}, we can interpret that $\mu^{(1|2)}_\csf$ 
represents an average channel flow from $S_1$ to $(\Shat_1,\Shat_2)$ given the side-information about $S_2.$
Similarly, $\mu^{(2|1)}_{\csf}$ represents
 an average channel flow  from $S_2$ given  information about $S_1$. 
%
\subsection{Synthesizing Converses for SW from Point-to-Point Duals}
As an immediate consequence of  Proposition~\ref{prop:sideconverse} and Proposition~\ref{prop:JEnDPSW}, we get that the point-to-point metaconverses in \eqref{eq:SIDmetaconverse}, \eqref{eq:SID2converse} and \eqref{eq:JEmetaconverse} are all lower bounds on $\OPT(\DPSW)$.
Consequently, the following is a straightforward lower bound on $\OPT(\DPSW)$,
\begin{align}
&\OPT(\DPSW)\geq \max \biggl \lbrace \sup_{ 0\leq \phih(s_1,s_2)\leq P_{S_1,S_2}(s_1,s_2)}\biggl\lbrace \sum_{s_1,s_2} \phih(s_1,s_2)\non \\&\qquad-M_1M_2\max_{\shat_1,\shat_2} \phih(\shat_1,\shat_2) \biggr \rbrace,\non\\
&\sup_{0\leq \phi^{(1|2)}(s_1,s_2) \leq P_{S_1,S_2}(s_1,s_2)}\hspace{-0.1cm}\biggl \lbrace \hspace{-0.05cm}\sum_{s_1,s_2}\phi^{(1|2)}(s_1,s_2)\non \\& \qquad -M_1 \sum_{s_2}\max_{\shat_1}\phi^{(1|2)}(\shat_1,s_2) \biggr\rbrace\non,\\
& \sup_{0\leq \phi^{(2|1)}(s_1,s_2) \leq P_{S_1,S_2}(s_1,s_2)}\biggl \lbrace \sum_{s_1,s_2}\phi^{(2|1)}(s_1,s_2)\non \\& \qquad -M_2 \sum_{s_1}\max_{\shat_2}\phi^{(2|1)}(s_1,\shat_2) \biggr\rbrace \biggr \rbrace.
\label{eq:maxconverse}
\end{align}

Convex analytically speaking, the above bound considers a \textit{convex combination} of feasible points of DPSW obtained via Propositions~\ref{prop:sideconverse} and \ref{prop:JEnDPSW}. 
In the following theorem we synthesize a new feasible point for DPSW by a \textit{nonlinear} combination of the source and channel flows in the point-to-point dual programs $\DPJE$, $\DPSI_{1|2}$ and $\DPSI_{2|1}$. We will subsequently apply specific metaconverses from Corollary~\ref{cor:metaJE}, Theorem~\ref{thm:metaside12} and Theorem~\ref{thm:metaside21} to get our new metaconverse.
\begin{theorem}\label{thm:feasDPSW}
Let $(\gammab^a,\gammab^b,\lambdab^{(1|2)}_{\ssf},\lambdab^{(1|2)}_{\csf}) \in \FEA(\DPSI_{1|2})$, $(\gammat^a,\gammat^b,\lambdat^{(2|1)}_{\ssf},\lambdat^{(2|1)}_{\csf})\in \FEA(\DPSI_{2|1})$ and $(\gammah^a,\gammah^b,\lambdah_{\ssf},\lambdah_{\csf}) \in \FEA(\DPJE)$. Then, any choice of values for the variables of DPSW satisfying the following equations is feasible for DPSW.
\begin{align}
&\lambda^{(1|2)}_{\ssf}(s_1,s_2,x_2,y_1,y_2,\shat_1,\shat_2)=\biggl[\lambdab^{(1|2)}_{\ssf}(s_1,s_2,\shat_1,y_1)\times\non\\&\Ibb\{x_2=y_2\} +\alpha \lambdah_{\ssf}(s_1,s_2,\shat_1,\shat_2,y_1,y_2)\biggr]\Ibb\{(s_1,s_2)=(\shat_1,\shat_2)\},\non
\\
&\lambda^{(2|1)}_{\ssf}(s_1,s_2,x_1,y_1,y_2,\shat_1,\shat_2)=\biggl[\lambdat^{(2|1)}_{\ssf}(s_1,s_2,\shat_2,y_2)\times\non\\&\Ibb\{x_1=y_1\} +(1-\alpha)\lambdah_{\ssf}(s_1,s_2,\shat_1,\shat_2,y_1,y_2)\biggr]\times \non\\& \hspace{3cm}\Ibb\{(s_1,s_2)=(\shat_1,\shat_2)\},\non \\
&\lambda_{\csf}(s_1,s_2,x_1,x_2,y_1,y_2)\hspace{-0.1cm}= \min \biggl \lbrace \hspace{-0.1cm} P(s_1,s_2)\Ibb\{(y_1,y_2)=(x_1,x_2)\},\non\\&\qquad\lambdah_{\csf}(s_1,s_2,x_1,x_2,y_1,y_2)+\lambdab^{(1|2)}_{\csf}(s_1,s_2,x_1,y_1)\Ibb\{x_2=y_2\}\non\\& \qquad +\lambdat^{(2|1)}_{\csf}(s_1,s_2,x_2,y_2)\Ibb\{x_1=y_1\}\biggr \rbrace,\label{eq:DPSWfeas}
\\
&\mu^{(2)}_{\csf}(s_2,x_2,y_1,y_2)\leq \biggl[\gammab^b(y_1,s_2) -\sum_{s_1\neq \shat_1}\lambdab^{(1|2)}_{\ssf}(s_1,s_2,\shat_1,y_2)\biggr]\non\\& \hspace{3cm}  \times \Ibb\{x_2=y_2\}\Ibb\{s_2=\shat_2\},\non
\end{align}
\begin{align}
&\mu^{(2)}_{\ssf}(s_2,\shat_1,\shat_2,y_1,y_2)= \alpha \lambdah_{\ssf}(\shat_1,s_2,\shat_1,\shat_2,y_1,y_2)\Ibb\{s_2=\shat_2\},\non
\\
&\mu^{(1)}_{\csf}(s_1,x_1,y_1,y_2)\leq \biggl[\gammat^b(s_1,y_2) -\sum_{s_2 \neq \shat_2}\lambdat^{(2|1)}_{\ssf}(s_2,s_2,\shat_2,y_2)\biggr]\non\\& \hspace{3cm} \times\Ibb\{x_1=y_1\}\Ibb\{s_1=\shat_1\},\non \\
&\mu^{(1)}_{\ssf}(s_1,\shat_1,\shat_2,y_1,y_2)= (1-\alpha)\lambdah_{\ssf}(s_1,\shat_2,\shat_1,\shat_2,y_1,y_2)\non\\&\hspace{3cm} \times\Ibb\{s_1=\shat_1\},\non
\\
&\mu^{(2|1)}_{\csf}(x_2,s_1,s_2)+\mu^{(1|2)}_{\csf}(x_1,s_1,s_2)\non\\&\qquad \leq \sum_{y_1,y_2}\lambda_{\csf}(s_1,s_2,x_1,x_2,y_1,y_2),\non
\end{align}where $\alpha \in (0,1)$ and
 $\gamma^a(s_1),\gamma^b(s_2)$ and $\gamma^c(y_1,y_2)$ are chosen such that (D1), (D2) and (D3) hold with equality.
\end{theorem}
\begin{proof}
The proof is included in Appendix \ref{sec:appB}.
\end{proof}

Theorem~\ref{thm:feasDPSW} generates a new feasible point for DPSW using an appropriate and nonlinear combination of the feasible points of the point-to-point source coding problems. 
Notice that the source flow
$\lambda^{(1|2)}_{\ssf}$ is taken as a superposition of the source flow for $\DPSI_{1|2}$ (which acts only when $x_2=y_2$) and a fraction of the source flow for DPJE. 
Similarly, $\lambda^{(2|1)}_{\ssf}$ is a superposition of the source flow for $\DPSI_{2|1}$ (which acts only when $x_1=y_1$) with the remaining fraction of the source flow for DPJE. 
While $\lambda_{\ssf}^{(1|2)}$ and $\lambda_{\ssf}^{(2|1)}$ are linear combinations of the point-to-point source flows, the channel flow $\lambda_{\csf}$ considers a nonlinear combination of the channel flows of $\DPSI_{1|2}$, $\DPSI_{2|1}$ and $\DPJE$,
so as to satisfy the bottleneck in constraint (D4). Since $\lambda_{\csf}$ dictates the choice of $\mu^{(1|2)}_{\csf}$ and $\mu_{\csf}^{(2|1)}$, the nonlinearity is also inherited in the relation of $\mu$'s. 

Moreover, the nonlinear relation between $\lambda_{\csf}$ and the channel flows of the point-to-point problems is one of the reasons for the
improvement on the classical converse of Miyake and Kanaya, as can be seen later in \eqref{eq:MKimprovement}. This improvement is otherwise hard to deduce from a feasible point of DPSW resulting from a convex combination of point-to-point feasible points.
 
 Thanks to Theorem~\ref{thm:feasDPSW}, to obtain finite blocklength converses for Slepian-Wolf coding, it suffices to consider the simpler  point-to-point source-coding problems and construct good feasible points for them.
In particular,
considering those feasible points of $\DPSI_{1|2}$, $\DPSI_{2|1}$ and $\DPJE$ which yield the metaconverses in \eqref{eq:SIDmetaconverse}, \eqref{eq:SID2converse} and \eqref{eq:JEmetaconverse} for the corresponding point-to-point sub-problems and subsequently employing Theorem~\ref{thm:feasDPSW}, we obtain the following new finite blocklength converse for SW.
 \begin{theorem}[Metaconverse for Slepian-Wolf Coding]\label{thm:SWmetaconverse}
 Consider the problem SW. Consequently, for any code, the following bound holds:
 \begin{align}
 &\Ebb[\Ibb\{(S_1,S_2)\neq (\Shat_1,\Shat_2)\}]\geq \OPT(\SW)\geq \OPT(\DPSW)\geq \non \\& \sup_{\phih,\phi^{(1|2)},\phi^{(2|1)}} \hspace{-0.1cm}\biggl \lbrace \sum_{s_1,s_2}\hspace{-0.1cm}\min\{P_{S_1,S_2}(s_1,s_2),\phih(s_1,s_2)\hspace{-0.05cm}+\hspace{-0.05cm}\phi^{(1|2)}(s_1,s_2)\non\\&+\phi^{(2|1)}(s_1,s_2)\}-M_1M_2\max_{\shat_1,\shat_2}\phih(\shat_1,\shat_2) \non\\&-M_2 \sum_{s_1}\max_{\shat_2}\phi^{(2|1)}(s_1,\shat_2) -M_1\sum_{s_2}\max_{\shat_1}\phi^{(1|2)}(\shat_1,s_2) \biggr \rbrace, \label{eq:SWmetaconverse}
 \end{align}
 where the supremum is over $\phih,\phi^{(1|2)},\phi^{(2|1)} :\Sscr_1\times \Sscr_2 \rightarrow [0,1]$ such that $0\leq \phih(s_1,s_2),\phi^{(1|2)}(s_1,s_2),\phi^{(2|1)}(s_1,s_2) \leq P_{S_1,S_2}(s_1,s_2)$ for all $s_1 \in \Sscr_1,s_2 \in \Sscr_2$.
 \end{theorem}
 \begin{proof} The proof is included in Appendix~\ref{sec:appB}. 
 \end{proof} 
 
 In particular, choosing 
\begin{align*}
 \phih(s_1,s_2)&= \min\{P_{S_1,S_2}(s_1,s_2),\eta_1(s_1,s_2)\}\\ 
 \phi^{(1|2)}(s_1,s_2)&= \min\{P_{S_1,S_2}(s_1,s_2),\eta_2(s_1,s_2)\}\\ 
 \phi^{(2|1)}(s_1,s_2)&= \min\{P_{S_1,S_2}(s_1,s_2),\eta_3(s_1,s_2)\}
\end{align*} 
in \eqref{eq:SWmetaconverse}, where $\eta_1,\eta_2,\eta_3 : \Sscr_1 \times \Sscr_2 \rightarrow [0,\infty)$, we get the following bound.
 \begin{align}
 &\Ebb[\Ibb\{(S_1,S_2)\neq (\Shat_1,\Shat_2)\}]\geq \OPT(\SW)\geq \OPT(\DPSW)\geq \non 
 \\& \sup_{\eta_1,\eta_2,\eta_3\geq 0} \biggl \lbrace \sum_{s_1,s_2}\min\{P_{S_1,S_2}(s_1,s_2),\eta_1(s_1,s_2)+\eta_2(s_1,s_2)+\non
 \\&\quad \eta_3(s_1,s_2)\}-M_1M_2\max_{\shat_1,\shat_2}\min\{P_{S_1,S_2}(\shat_1,\shat_2),\eta_1(\shat_1,\shat_2)\}\non\\&\quad  -M_2 \sum_{s_1}\max_{\shat_2}\min\{P_{S_1,S_2}(s_1,\shat_2),\eta_3(s_1,\shat_2)\}\non\\&\quad -M_1\sum_{s_2}\max_{\shat_1}\min\{P_{S_1,S_2}(\shat_1,s_2),\eta_2(\shat_1,s_2)\} \biggr \rbrace \label{eq:SWmetaconverse1}
 \end{align}
Further, the new converse in \eqref{eq:SWmetaconverse1} improves on the information spectrum based converse of Miyake and Kanaya \cite{miyake1995coding} as shown in the following corollary. 
\begin{corollary}[Improvements on Miyake-Kanaya Converse]
The converse in \eqref{eq:SWmetaconverse1} implies the following improvement on the converse of Miyake and Kanaya,
\begin{align}
&\Ebb[\I{(S_1,S_2)\neq (\Shat_1,\Shat_2)}]\geq \OPT({\rm SC})\geq \OPT({\rm DP})\geq\non\\& \sup_{\beta>0}\biggl\lbrace \Pbb\biggl[h_{S_1,S_2}(S_1,S_2)\geq \log M_1M_2+\beta \hspace{0.2cm} \mbox{or}\hspace{0.2cm} h_{S_1|S_2}(S_1|S_2)\non\\& \geq \log M_1+\beta \hspace{0.1cm}\mbox{or}\hspace{0.2cm} h_{S_2|S_1}(S_2|S_1)\geq \log M_2+\beta \biggr]+\non\\
& \sum_{s_1,s_2}\max\biggl \lbrace \frac{\exp(-
 \beta)}{M_1M_2},\frac{\exp(-\beta)}{M_1}P_{S_2}(s_2),
 \frac{\exp(-\beta)}{M_2}P_{S_1}(s_1)\biggr \rbrace \non \\&\times \Ibb\biggl \lbrace P_{S_1|S_2}(s_1|s_2)>\frac{\exp(-\beta)}{M_1},P_{S_2|S_1}(s_2|s_1)>\frac{\exp(-\beta)}{M_2},\non\\&\quad P_{S_1,S_2}(s_1,s_2)>\frac{\exp(-\beta)}{M_1M_2}\biggr \rbrace
  -3\exp(-\beta)\biggr\rbrace, \label{eq:MKimprovement}
  \end{align} where  $h_{A|B}(a|b)\equiv -\log P_{A|B}(a|b)$ is the conditional entropy density and  $h_{A,B}(a,b)\equiv -\log P_{A,B}(a,b)$ is the joint entropy density.  \end{corollary}
  \begin{proof}
  To obtain the above converse, weaken \eqref{eq:SWmetaconverse1} by choosing $\eta_1(s_1,s_2)\equiv \frac{\exp(-\beta)}{M_1M_2},$ $\eta_2(s_1,s_2)\equiv P_{S_2}(s_2)\frac{\exp(-\beta)}{M_1}$, $\eta_3(s_1,s_2)\equiv P_{S_1}(s_1)\frac{\exp(-\beta)}{M_2}$ and bound  $ \min\{P(s_1,s_2),\eta_1(s_1,s_2)+\eta_2(s_1,s_2)+\eta_3(s_1,s_2)\}$ by $ \min\{P(s_1,s_2),\max\{\eta_1(s_1,s_2),\eta_2(s_1,s_2),\eta_3(s_1,s_2)\}\}.$ Further, bound $\min\{P_{S_1,S_2}(\shat_1,\shat_2),\eta_1(\shat_1,\shat_2)\}$ by $\eta_1(s_1,s_2)$, $\min\{P_{S_1,S_2}(\shat_1,\shat_2),\eta_2(\shat_1,\shat_2)\}$ by $\eta_2(s_1,s_2)$ and $\min\{P_{S_1,S_2}(\shat_1,\shat_2),\eta_3(\shat_1,\shat_2)\}$ by $\eta_3(s_1,s_2)$.
  Subsequently, employing the definition of  $h_{A|B}(a|b)$,  $h_{A,B}(a,b)$ and taking supremum over $\beta>0$, we get the required converse.
  \end{proof}
  \begin{remarkc}[(Recovering the Converse of Miyake and Kanaya)]
  Lower bounding the non-negative term in \eqref{eq:MKimprovement} corresponding to $\Ibb\biggl \lbrace P_{S_1|S_2}(s_1|s_2)>\frac{\exp(-\beta)}{M_1},P_{S_2|S_1}(s_2|s_1)>\frac{\exp(-\beta)}{M_2},P_{S_1,S_2}(s_1,s_2)>\frac{\exp(-\beta)}{M_1M_2}\biggr \rbrace$ with zero, we recover the converse of Miyake and Kanaya given as,
\begin{align}
&\Ebb[\I{(S_1,S_2)\neq (\Shat_1,\Shat_2)}]\geq \sup_{\beta>0}\biggl\lbrace \Pbb\biggl[h_{S_1,S_2}(S_1,S_2)\non\\&\geq \log M_1M_2+\beta \hspace{0.2cm} \mbox{or} \hspace{0.2cm} h_{S_1|S_2}(S_1|S_2)\geq \log M_1+\beta \non\\& \hspace{0.1cm}\mbox{or}\hspace{0.2cm} h_{S_2|S_1}(S_2|S_1)\geq \log M_2+\beta \biggr]
  -3\exp(-\beta)\biggr\rbrace. \label{eq:MKconverse}
  \end{align}
  \end{remarkc}\\

Before we conclude, we note that the relevance of \eqref{eq:maxconverse}, particularly in the analysis of second-order asymptotics  is limited. As pointed out in \cite[Section 6.2]{tan2014asymptotic}, the second-order analysis centered at a corner point of the first order rate region of Slepian-Wolf problem, is determined by the multivariate Gaussian CDF with respect to jointly encoded and side-information  problems together. Consequently, the lower bound in \eqref{eq:SWmetaconverse} or the union bound in \eqref{eq:MKimprovement} are more relevant. In fact, with the flexibility of choosing $\eta_1,\eta_2,\eta_3$ which are functions of $(s_1,s_2)$, the converse in \eqref{eq:SWmetaconverse} may even yield refined third order terms in the asymptotic analysis.
	\subsection{Illustrative Example:Doubly Symmetric Binary Sources}	
	In this section, we consider the example of a Doubly Symmetric Binary Source (DSBS) with $\Sscr_1=\Sscr_2=\{0,1\}^n$ and the joint probability distribution given as,
	\begin{align}P_{S_1,S_2}(s_1,s_2)\equiv \frac{1}{2^n}p^{d(s_1,s_2)}(1-p)^{n-d(s_1,s_2)},\label{eq:DSBS}\end{align} where $p<0.5$ and $d(s_1,s_2)$ represents the Hamming distance between $s_1 \in \Sscr_1$ and $s_2 \in \Sscr_2$. Further,  $M_1=2^{nR_1}$, $M_2=2^{nR_2},$  $P_{S_1|S_2}(s_1|s_2)=P_{S_2|S_1}(s_2|s_1)\equiv p^{d(s_1,s_2)}(1-p)^{n-d(s_1,s_2)}$ and $P_{S_1}(s_1)=P_{S_2}(s_2)=\frac{1}{2^n}$. The optimal rate region for this DSBS is given by \cite{el2011network},
	\begin{align*}
	\mathcal{R}_{SW}=\{(R_1,R_2)\mid R_1,R_2 \geq H(p), R_1+R_2 \geq 1+H(p)\}.
	\end{align*}
	
Particularizing the converse in \eqref{eq:SWmetaconverse1} to the case of DSBS by choosing $\eta_2(s_1,s_2)\equiv P_{S_2}(s_2)\frac{2^{-\beta}}{M_1}$, $\eta_3(s_1,s_2)\equiv P_{S_1}(s_1)\frac{2^{-\beta}}{M_2}$ and $\eta_1(s_1,s_2)\equiv \frac{2^{-\beta}}{M_1M_2}$	results in the following converse.
\begin{align}
&\Ebb[\Ibb\{(S_1,S_2) \neq (\Shat_1,\Shat_2)\}]\geq  \OPT(\DPSW)\non\\&\geq \sup_{\beta>0}\biggl \lbrace \sum_{k=0}^n \Comb{n}{k} \min \biggl \lbrace p^k(1-p)^{n-k}, \frac{2^{-\beta}}{M_1}+\frac{2^{-\beta}}{ M_2}+\frac{2^{-\beta+n}}{M_2 M_1}\biggr \rbrace \non \\&\quad  -M_12^n \max_{k \in \{0,\hdots n\}}\min \biggl \lbrace \frac{p^k(1-p)^{n-k}}{2^n}, \frac{2^{-\beta}}{2^nM_1} \biggr \rbrace \non 
\\&\quad  -M_22^n \max_{k \in \{0,\hdots n\}}\min \biggl \lbrace \frac{p^k(1-p)^{n-k}}{2^n}, \frac{2^{-\beta}}{2^nM_2} \biggr \rbrace \non \\
& -M_1M_2 \max_{k \in \{0,\hdots n\}}\min \biggl \lbrace \frac{p^k(1-p)^{n-k}}{2^n}, \frac{2^{-\beta}}{M_1M_2} \biggr \rbrace \biggr \rbrace. \label{eq:DSBSmeta}
\end{align}The above bound follows from \eqref{eq:SWmetaconverse} since for any $s_1 \in \Sscr_1$, the number of $s_2$'s at a Hamming distance of $k \in \{0,\hdots,n\}$ is given by $\Comb{n}{k}$.

 Figure~\ref{fig:SWout} and Figure~\ref{fig:SWins} compare our improved converse \eqref{eq:DSBSmeta} with the Miyake-Kanaya converse in \eqref{eq:MKconverse}. It is seen see that the improved converse in \eqref{eq:DSBSmeta} shows a nontrivial improvement on the Miyake-Kanaya converse. Note the differences in scale in both Figure~\ref{fig:SWout} and Figure~\ref{fig:SWins}.\\
\begin{remarkc}
For the case of DSBS whose joint distribution depends only on the Hamming distance, choosing $\eta_1,\eta_2,\eta_3$ in \eqref{eq:SWmetaconverse1} to  be independent of $(s_1,s_2)$ results in  \eqref{eq:DSBSmeta} performing weaker than \eqref{eq:maxconverse}. To see this, lower bound \eqref{eq:maxconverse} with the converse from jointly encoded sources \ie, \eqref{eq:JEmetaconverse}, and  choose 
\begin{align}&\phih(s_1,s_2)=\min\biggl \lbrace p^{d(s_1,s_2)}(1-p)^{n-d(s_1,s_2)}\frac{1}{2^n},\non \\&\min\biggl \lbrace p^{d(s_1,s_2)}(1-p)^{n-d(s_1,s_2)}\frac{1}{2^n}, \frac{2^{-\beta}}{2^n M_1}\biggr \rbrace\non\\&+\min\biggl \lbrace p^{d(s_1,s_2)}(1-p)^{n-d(s_1,s_2)}\frac{1}{2^n}, \frac{2^{-\beta}}{2^n M_2}\biggr \rbrace \non\\&+\min\biggl \lbrace p^{d(s_1,s_2)}(1-p)^{n-d(s_1,s_2)}\frac{1}{2^n}, \frac{2^{-\beta}}{M_1M_2}\biggr \rbrace \biggr \rbrace\label{eq:phih}
\end{align}
and  subsequently, upper bound $\max_{\shat_1,\shat_2} \phih(\shat_1,\shat_2)$ with 
\begin{align*}&\max_{\shat_1,\shat_2}\biggl \lbrace \min\{p^{d(\shat_1,\shat_2)}(1-p)^{n-d(\shat_1,\shat_2)}\frac{1}{2^n}, \frac{2^{-\beta}}{2^nM_1}\}+\\& \min\{p^{d(\shat_1,\shat_2)}(1-p)^{n-d(\shat_1,\shat_2)}\frac{1}{2^n},\frac{2^{-\beta}}{2^nM_2}\}\\& +\min\{p^{d(\shat_1,\shat_2)}(1-p)^{n-d(\shat_1,\shat_2)}\frac{1}{2^n},\frac{2^{-\beta}}{M_2M_1}\}\biggr \rbrace.\end{align*} 
Further, noting that $\phih(s_1,s_2)$ is equivalent to,
$\min\{p^{d(s_1,s_2)}(1-p)^{n-d(s_1,s_2)}\frac{1}{2^n},\frac{2^{-\beta}}{M_12^n}+\frac{2^{-\beta}}{ 2^n M_2}+\frac{2^{-\beta}}{M_2 M_1}\}$ and $\sum_{s_1,s_2}\phih(s_1,s_2)= \sum_{k=0}^n\Comb{n}{k} \min\biggl \lbrace p^{k}(1-p)^{n-k},\frac{2^{-\beta}}{M_1}+\frac{2^{-\beta}}{  M_2}+\frac{2^{-\beta+n}}{M_2 M_1}\}$ the following lower bound follows from \eqref{eq:maxconverse},
\begin{align}
&\sum_{k=0}^n\Comb{n}{k} \min\biggl \lbrace p^{k}(1-p)^{n-k},\frac{2^{-\beta}}{M_1}+\frac{2^{-\beta}}{  M_2}+\frac{2^{-\beta+n}}{M_2 M_1}\biggr \rbrace \non\\&-
M_1M_2\max_{k \in \{0,\hdots,n\}}\biggl \lbrace \min\{p^{k}(1-p)^{n-k}\frac{1}{2^n}, \frac{2^{-\beta}}{2^nM_1}\}+ \non\\& \min\{p^{k}(1-p)^{n-k}\frac{1}{2^n},\frac{2^{-\beta}}{2^nM_2}\}+\min\{p^{k}(1-p)^{n-k}\frac{1}{2^n},\non\\&\qquad\frac{2^{-\beta}}{M_2M_1}\}\biggr \rbrace \non\\
&\geq \sum_{k=0}^n\Comb{n}{k} \min\biggl \lbrace p^{k}(1-p)^{n-k},\frac{2^{-\beta}}{M_1}+\frac{2^{-\beta}}{  M_2}+\frac{2^{-\beta+n}}{M_2 M_1}\biggr \rbrace \non\non\\&-
M_1M_2\max_{k \in \{0,\hdots,n\}} \min\{p^{k}(1-p)^{n-k}\frac{1}{2^n}, \frac{2^{-\beta}}{2^nM_1}\}\non\\& -M_1M_2\max_{k \in \{0,\hdots,n\}}\min\{p^{k}(1-p)^{n-k}\frac{1}{2^n},\frac{2^{-\beta}}{2^nM_2}\}\non\\&-M_1M_2\max_{k \in \{0,\hdots,n\}}\min\{p^{k}(1-p)^{n-k}\frac{1}{2^n},\frac{2^{-\beta}}{M_2M_1}\}\biggr \rbrace. \label{eq:bound2}
\end{align}
It is now easy to see that when $M_1,M_2 \leq 2^n$, \eqref{eq:bound2} outperforms \eqref{eq:DSBSmeta}. 
Note, however, that this outperformance relies on a particular choice of the $\eta_1,\eta_2,\eta_3$ in Theorem~\ref{thm:SWmetaconverse} which leads to \eqref{eq:DSBSmeta} and on the structure of the DSBS. In particular, it \textit{does not} imply that Theorem~\ref{thm:SWmetaconverse} is weaker than \eqref{eq:maxconverse}.
\end{remarkc}\\
	\begin{figure}[htb]
	\centering
\includegraphics[scale=0.843,clip=true,trim=2.15in 5.4in 0in 3.4in]{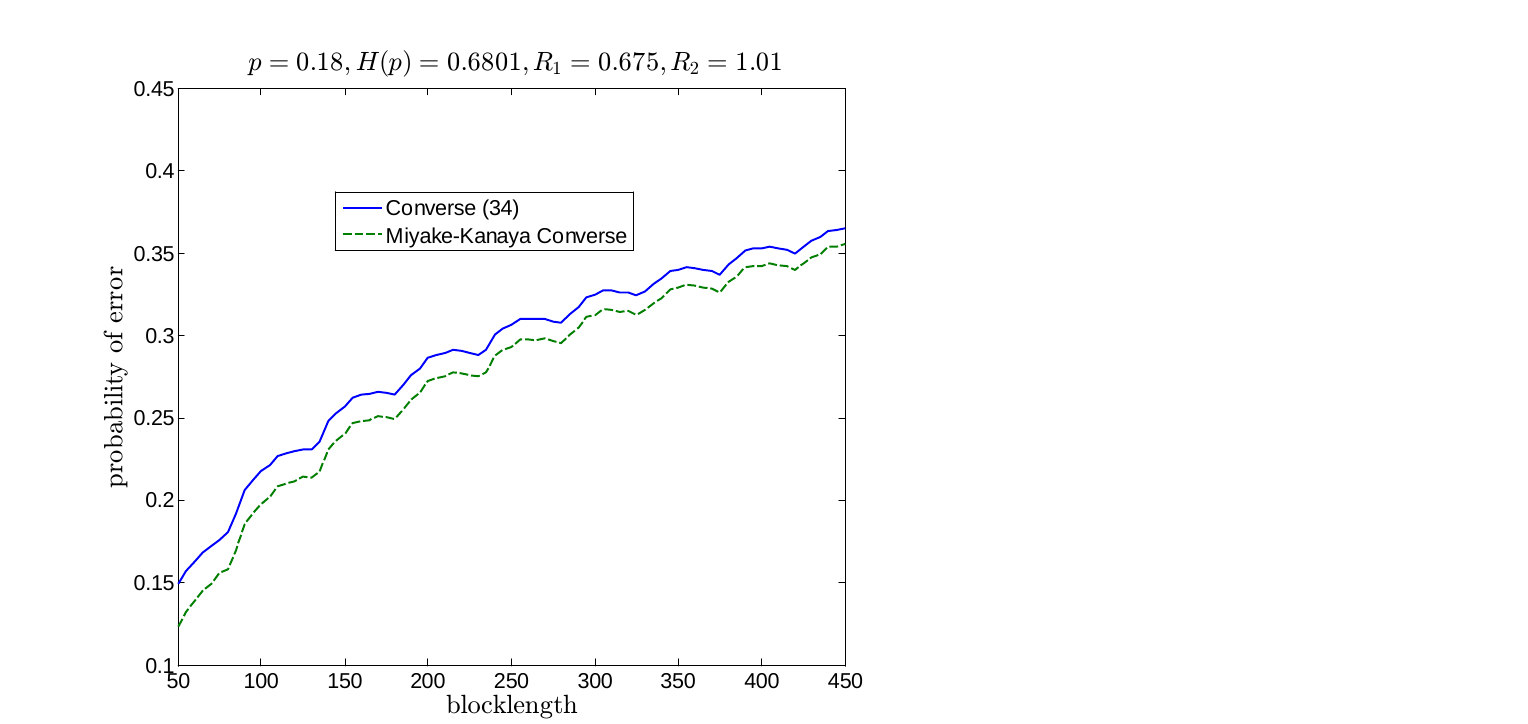}
\caption{DSBS with $(R_1,R_2) \not \in \mathcal{R}_{SW}$.}
\label{fig:SWout}
\vspace{0.1cm}
\includegraphics[scale=0.8,clip=true,trim=2.3in 5.1in 0in 3.5in]{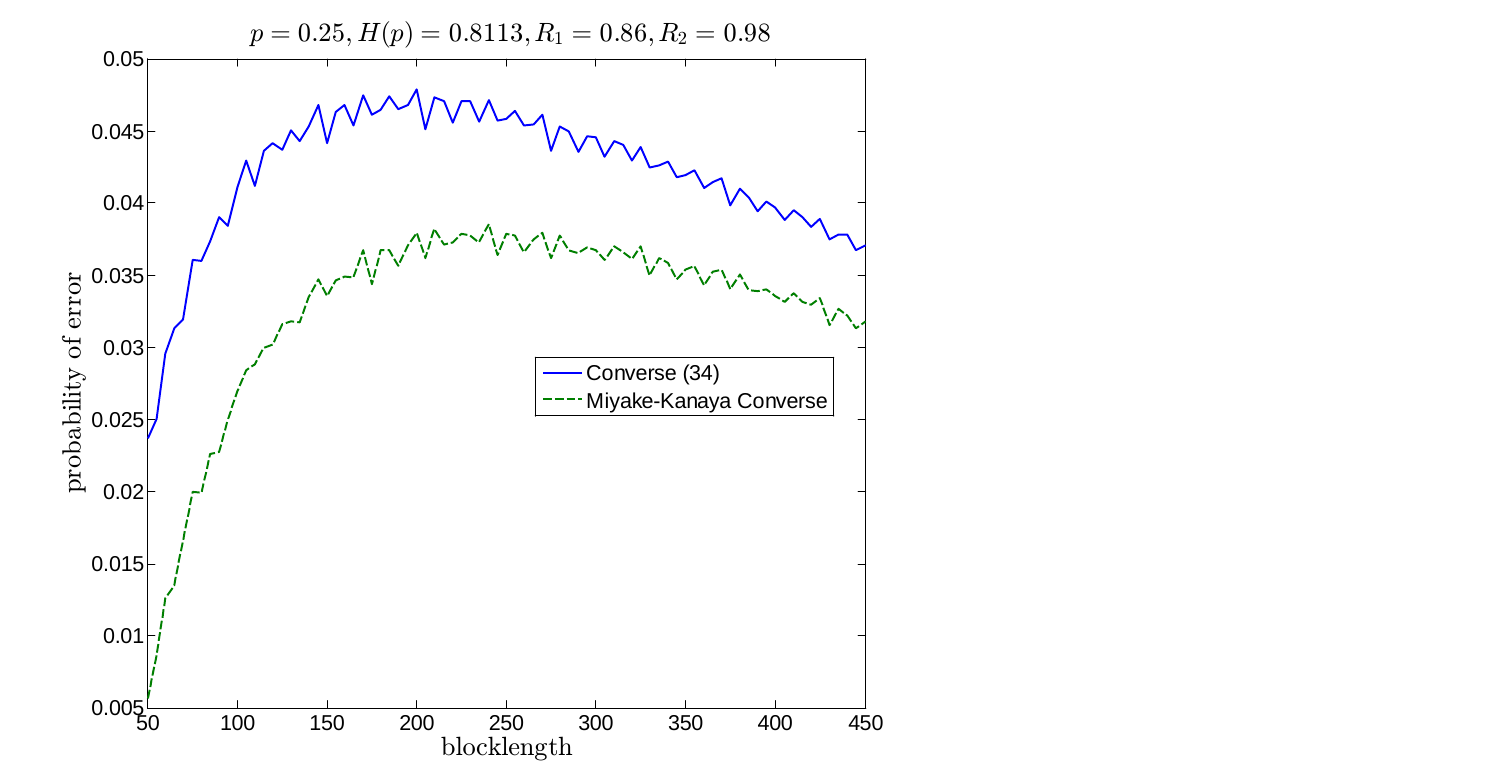}
\caption{DSBS with $(R_1,R_2) \in \mathcal{R}_{SW}$.}
\label{fig:SWins}
\end{figure}
\section{Discussion}\label{sec:interpretation}
\begin{figure}
\tikzstyle{start} = [rectangle,minimum width=0.5cm, minimum height=0.3cm,text centered, draw=black, fill=white!]
\tikzstyle{arrow} = [thick,->,>=stealth]
\begin{small}
\begin{center}
\begin{tikzpicture}[node distance=1.1cm]
\node (T1) [start] {$\OPT(\DPSW)$};
\node (T2)[start, below =of T1] {Converse \eqref{eq:avenue11}};
\node(T3) [start,below =of T2] {Converse \eqref{eq:avenue1}};
    \node(T4) [start,below =of T3] {Converse \eqref{eq:level3bound}};
    \node(T5) [start,right =of T3] {Metaconverse \eqref{eq:SWmetaconverse}};
    \node (T6) [start,right =of T4] {Miyake-Kanaya Converse \eqref{eq:MKconverse}};
 \draw [arrow] (T1) -- node[right] { {\small $\min (\rm{sum} 's)\geq \rm{sum}(\min 's)$ as in \eqref{eq:reason1} }} (T2);
\draw [arrow] (T2) -- node[right] { \small $\min (\rm{sum} 's)\geq \rm{sum}(\min 's)$ as in \eqref{eq:reason2}} (T3);
\draw [arrow] (T3) -- node[right] { $\substack{\sum_{s_1}\lambda_{\ssf}^{(1|2)} \independent \shat_1,\\ \sum_{s_2}\lambda_{\ssf}^{(2|1)}\independent \shat_2}$} (T4);
\draw [arrow] (T3)-- node[above] { \small $\substack{ \rm{dual variables} \\ \rm{in \hspace{0.1cm}\eqref{eq:constr6}}}$} (T5);
\draw [arrow] (T4) -- node[below] { $\substack{\lambda_{\ssf}^{(1|2)},\\ \lambda_{\ssf}^{(2|1)}\\  \rm{as \hspace{0.1cm} in \hspace{0.1cm}} \eqref{eq:mkvariables}}$} (T6);
              \draw [->,>=stealth,thick] (T5.south) to  (T6.north);
\end{tikzpicture}
\end{center}
\caption{Hierarchy of lower bounds derived. An arrow from $A \rightarrow B$ implies that $A \geq B$, the heading above the arrow indicate how $B$ is obtained from $A$. $\min (\rm{sum} 's)\geq \rm{sum} (\min 's)$ represent that minimum of sums is atleast equal to the sum of minimums.
\vspace{-1cm}}\label{fig:hierarchy}
\end{small} 
\end{figure}
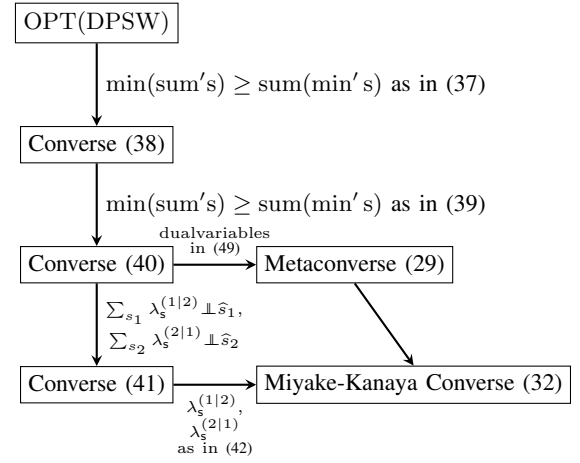
The tightest finite blocklength converse derivable for the SW problem employing the LP-based framework is $\OPT(\DPSW)$, the exact evaluation of which is difficult. However,  since the framework poses the tightest converse as an optimization problem, a hierarchy of lower bounds on it can be derived through a series of optimization problems. This  also helps us conceptually situate the metaconverse in \eqref{eq:SWmetaconverse} and the Miyake-Kanaya converse in the hierarchy, as discussed below.

Recall that $\OPT(\DPSW)$ evaluates to the following optimization problem,
\begin{align*}
&\max_{\bar{\Theta}}\biggl \lbrace \hspace{-0.1cm} \sum_{s_1}\min_{x_1}\biggl \lbrace\sum_{y_1,y_2}\hspace{-0.1cm} \mu^{(1)}_{\csf}(s_1,x_1,y_1,y_2)\hspace{-0.1cm}+\hspace{-0.1cm}\sum_{s_2}\mu^{(1|2)}_{\csf}(x_1,s_1,s_2)\biggr \rbrace \non \non\\& + \sum_{s_2} \min_{x_2} \biggl \lbrace \sum_{y_1,y_2} \mu^{(2)}_{\csf}(s_2,x_2,y_1,y_2) +\sum_{s_1}\mu^{(2|1)}_{\csf}(x_2,s_1,s_2)\biggr \rbrace  \non \\&+\sum_{y_1,y_2} \min_{\shat_1,\shat_2}\biggl \lbrace \hspace{-0.1cm}\sum_{s_2}\mu^{(2)}_{\ssf}(s_2,\shat_1,\shat_2,y_1,y_2) \non\\&+\sum_{s_1}\mu^{(1)}_{\ssf}(s_1,\shat_1,\shat_2,y)\biggr \rbrace \biggr \rbrace,\\
&\mbox{s.t.}\qquad \mbox{(D4)-(D7)} \hspace{0.2cm}\mbox{hold.}
\end{align*}Note that in the above equation, the source flow $\mu_{\ssf}^{(1)}$ and channel flow $\mu_{\csf}^{(1)}$, though along the same path are under two different minimum's. So are the pairs $(\mu_{\ssf}^{(2)},\mu_{\csf}^{(2)})$ and $(\mu_{\csf}^{(1|2)},\mu_{\csf}^{(2|1)})$. Consequently, we first try to bring the terms in each of the above pairs together. Towards this, we first separate out the terms inside the minimum's in $\OPT(\DPSW)$ by employing that
\begin{align}
\sum_{s_1}\min_{x_1}\biggl \lbrace\sum_{y_1,y_2} \mu^{(1)}_{\csf}+\sum_{s_2}\mu^{(1|2)}_{\csf}\biggr \rbrace &\geq \sum_{s_1}\biggl[\min_{x_1}\sum_{y_1,y_2} \mu^{(1)}_{\csf}\non\\&+\min_{x_1}\sum_{s_2}\mu^{(1|2)}_{\csf}\biggr ] \non,\\
\sum_{s_2} \min_{x_2} \biggl \lbrace \sum_{y_1,y_2} \mu^{(2)}_{\csf} +\sum_{s_1}\mu^{(2|1)}_{\csf}\biggr \rbrace& \geq \sum_{s_2} \biggl [ \min_{x_2}\sum_{y_1,y_2} \mu^{(2)}_{\csf}\non \\&+\min_{x_2}\sum_{s_1}\mu^{(2|1)}_{\csf} \biggr],\non\\
\sum_{y_1,y_2} \min_{\shat_1,\shat_2}\biggl \lbrace \sum_{s_2}\mu^{(2)}_{\ssf} +\sum_{s_1}\mu^{(1)}_{\ssf}\biggr \rbrace& \geq \sum_{y_1,y_2} \biggl[ \min_{\shat_1,\shat_2}\sum_{s_2}\mu^{(2)}_{\ssf}\non \\&+\min_{\shat_1,\shat_2}\sum_{s_1}\mu^{(1)}_{\ssf}\biggr ]\label{eq:reason1}.
\end{align}This results in the following optimization problem, the optimal value of which yields a lower bound on $\OPT(\DPSW)$,

\begin{align}
&\OPT(\DPSW)\geq \max_{\bar{\Theta}} \biggl \lbrace \sum_{s_1}\min_{x_1}\sum_{y_1,y_2} \mu^{(1)}_{\csf}(s_1,x_1,y_1,y_2)+\non \\&\sum_{s_1}\min_{x_1}\sum_{s_2}\mu^{(1|2)}_{\csf}(x_1,s_1,s_2)\hspace{-0.1cm} +\hspace{-0.1cm} \sum_{s_2} \min_{x_2}  \sum_{y_1,y_2} \mu^{(2)}_{\csf}(s_2,x_2,y_1,y_2)\non \\&\qquad +\sum_{s_2} \min_{x_2}\sum_{s_1}\mu^{(2|1)}_{\csf}(x_2,s_1,s_2)\non\\&\qquad +\sum_{y_1,y_2} \min_{\shat_1,\shat_2} \sum_{s_2}\mu^{(2)}_{\ssf}(s_2,\shat_1,\shat_2,y_1,y_2)\non\\& \qquad +\sum_{y_1,y_2} \min_{\shat_1,\shat_2}\sum_{s_1}\mu^{(1)}_{\ssf}(s_1,\shat_1,\shat_2,y)\biggr \rbrace \biggr \rbrace\non \\& \mbox{s.t.}\qquad \mbox{(D4)-(D7)} \hspace{0.2cm}\mbox{hold.} \label{eq:avenue11}
\end{align}
We now further lower bound \eqref{eq:avenue11} by using,
\begin{align}
\min_{x_1}\sum_{y_1,y_2}\mu_{\csf}^{(1)} &\geq \sum_{y_2}\min_{x_1} \sum_{y_1}\mu_{\csf}^{(1)} \non\\
\min_{x_2}\sum_{y_1,y_2}\mu_{\csf}^{(2)} &\geq \sum_{y_1}\min_{x_2} \sum_{y_2}\mu_{\csf}^{(2)}, \label{eq:reason2}
\end{align}to get the following optimization problem, whose optimal value is a lower bound on \eqref{eq:avenue11},
\begin{align}
&\max_{\substack{\lambda^{(2|1)}_{\ssf},\lambda^{(1|2)}_{\ssf},\lambda_{\csf} \\ \tiny \mbox{s.t (D4) holds}} }\Biggl \lbrace \max_{\substack{\mu^{(2)}_{\ssf},\mu^{(2)}_{\csf}\\ \tiny \mbox{s.t (D5) holds}}}\biggl \lbrace \sum_{y_1}\biggl[ \sum_{s_2} \min_{x_2}  \sum_{y_2} \mu^{(2)}_{\csf}(s_2,x_2,y_1,y_2)\non \\& \quad +\sum_{y_2} \min_{\shat_1,\shat_2}\sum_{s_2}\mu^{(2)}_{\ssf}(s_2,\shat_1,\shat_2,y_1,y_2)\biggr] \biggr \rbrace+ \label{eq:avenue1} \\
&  \max_{\substack{\mu^{(1)}_{\ssf},\mu^{(1)}_{\csf} \\ \tiny \mbox{s.t (D6) holds}}} \biggl \lbrace \sum_{y_2}\biggl[ \sum_{s_1}\min_{x_1}\sum_{y_1} \mu^{(1)}_{\csf}(s_1,x_1,y_1,y_2)\non \\&\quad +\sum_{y_1} \min_{\shat_1,\shat_2}\sum_{s_1}\mu^{(1)}_{\ssf}(s_1,\shat_1,\shat_2,y_1,y_2)\biggr]\biggr \rbrace +\non \\
&\max_{\substack{\mu^{(2|1)}_{\csf},\mu^{(1|2)}_{\csf}\\ \tiny \mbox{s.t (D7) holds}} }\biggl \lbrace \sum_{s_1}\min_{x_1}\sum_{s_2}\mu^{(1|2)}_{\csf}(x_1,s_1,s_2)\non
\end{align}
\begin{align}
&\qquad+\sum_{s_2} \min_{x_2} \sum_{s_1}\mu^{(2|1)}_{\csf}(x_2,s_1,s_2)\ \biggr \rbrace \Biggr \rbrace. \non
\end{align} 
Note that \eqref{eq:avenue1} now has an outer optimization over $\lambda_{\ssf}^{(1|2)}, \lambda_{\ssf}^{(2|1)},\lambda_{\csf}$ satisfying the error density bottleneck (D4) and three inner optimization problems over each of the pairs, $(\mu_{\ssf}^{(1)},\mu_{\csf}^{(1)})$, $(\mu_{\ssf}^{(2)},\mu_{\csf}^{(2)})$ and $(\mu_{\csf}^{(1|2)},\mu_{\csf}^{(2|1)})$ with bottlenecks imposed by constraints (D6), (D5) and (D7) respectively.

We further lower bound  \eqref{eq:avenue1} by restricting the choice of $\lambda_{\ssf}^{(1|2)}$ and $\lambda_{\ssf}^{(2|1)}$ such that $\sum_{s_1}\lambda_{\ssf}^{(1|2)}$ is independent of $\shat_1$ and $\sum_{s_2}\lambda_{\ssf}^{(2|1)}$ is independent of $\shat_2$. Under this assumption, constraints (D5) and (D6)  imply that $\mu_{\ssf}^
{(2)}$ and $\mu_{\ssf}^{(1)}$  are independent of $\shat_1$ and $\shat_2$, respectively. Hence, for each $y_1 \in \Yscr_1$, $\sum_{s_2} \min_{x_2}  \sum_{y_2} \mu^{(2)}_{\csf}(s_2,x_2,y_1,y_2)+\sum_{y_2} \min_{\shat_1,\shat_2}\sum_{s_2}\mu^{(2)}_{\ssf}(s_2,\shat_1,\shat_2,y_1,y_2)$ represents the objective corresponding to the packing of source flow $\mu^{(2)}_{\ssf}$ and channel flow $\mu^{(2)}_{\csf}$ through the path $S_2\rightarrow X_2 \rightarrow Y_2 \rightarrow \Shat_2$ satisfying the bottleneck, $\mu^{(2)}_{\ssf}(s_2,\shat_1,\shat_2,y_1,y_2)+\mu^{(2)}_{\csf}(s_2,x_2,y_1,y_2) \leq \sum_{s_1} \lambda_{\ssf}^{(1|2)}(s_1,s_2,x_2,y_1,y_2,\shat_1,\shat_2)$ for all $s_2,x_2,y_2,\shat_2$.  Taking the maximum over $\mu_{\ssf}^{(2)}, \mu_{\csf}^{(2)}$ inside the summation over $y_1$, we can express the optimal packing of these flows as $\OPT(\DP, \sum_{s_1} \lambda_{\ssf}^{(1|2)}(s_1,s_2,x_2,y_1,y_2,\shat_1,\shat_2))$, defined as in \eqref{eq:optDP}. Note that here, the RHS of the bottleneck is not necessarily an error density, but a function of $(s_2,x_2,y_1,y_2,\shat_2)$ and $\OPT(\DP, \sum_{s_1} \lambda_{\ssf}^{(1|2)}(s_1,s_2,x_2,y_1,y_2,\shat_1,\shat_2))$ is not necessarily the relaxation of a source coding problem.

Similarly, for each $y_2 \in \Yscr_2$, $\sum_{s_1}\min_{x_1}\hspace{-0.1cm}\sum_{y_1} \mu^{(1)}_{\csf}(s_1,x_1,y_1,y_2)+\sum_{y_1} \min_{\shat_1,\shat_2}\sum_{s_1}\hspace{-0.1cm}\mu^{(1)}_{\ssf}(s_1,\shat_1,\shat_2,y_1,y_2)$ represents the objective corresponding to the packing of source flow $\mu^{(1)}_{\ssf}$ and the channel flow $\mu^{(1)}_{\csf}$ through the path $S_1 \rightarrow X_1 \rightarrow Y_1 \rightarrow \Sscrhat_1$ satisfying the bottleneck imposed by (D6). The resultant optimal packing can be expressed as $\OPT(\DP, \sum_{s_2} \lambda_{\ssf}^{(2|1)}(s_1,s_2,x_1,y_1,y_2,\shat_1,\shat_2))$.
Employing these yields the following lower bound on \eqref{eq:avenue1},

\begin{align}
&\max_{\substack{\lambda^{(1|2)}_{\ssf},\lambda^{(2|1)}_{\ssf},\lambda_{\csf}\\ \tiny \mbox{s.t (D4) holds} \\ \sum_{s_1} \lambda^{(1|2)}_{\ssf} \independent \shat_1, \\ \sum_{s_2}\lambda^{(2|1)}_{\ssf} \independent \shat_2 }}\hspace{-0.12cm}\Biggl \lbrace \hspace{-0.12cm}\sum_{y_1}\hspace{-0.12cm}\OPT(\DP,\sum_{s_1}\lambda^{(1|2)}_{\ssf}(s_1,s_2,x_2,y_1,y_2,\shat_1,\shat_2))\non\\ &
 + \sum_{y_2}\OPT(\DP,\sum_{s_2}\lambda^{(2|1)}_{\ssf}(s_1,s_2, x_1,y_1,y_2,\shat_1,\shat_2))\non\\& +\max_{\substack{\mu^{(1|2)}_{\csf},\mu^{(2|1)}_{\csf} \\ \tiny \mbox{s.t (D7) holds}} }\biggl \lbrace \sum_{s_1}\min_{x_1}\sum_{s_2}\mu^{(1|2)}_{\csf}(x_1,s_1,s_2)\non \\&\hspace{1cm}+\sum_{s_2} \min_{x_2} \sum_{s_1}\mu^{(2|1)}_{\csf}(x_2,s_1,s_2)\ \biggr \rbrace \Biggr \rbrace \label{eq:level3bound}.
\end{align} 
Note that for a given choice of $\lambda_{\ssf}^{(1|2)}, \lambda_{\ssf}^{(2|1)}$ the bound in \eqref{eq:level3bound} comprises of 
optimal value of the duals of point-to-point problems, $\OPT(\DP,\sum_{s_1}\lambda_{\ssf}^{(1|2)})$ and $\OPT(\DP,\sum_{s_2}\lambda_{\ssf}^{(2|1)})$.
However, the objective of these problems is necessarily source coding since RHS of (D5) (or (D6)) is not the source coding error density.
Thus, the bounds in \eqref{eq:avenue11}, \eqref{eq:avenue1} and \eqref{eq:level3bound} illustrate a hierarchy of lower bounds on the optimal value of DPSW. We now show that the Miyake and Kanaya converse falls lower in this hierarchy than our converse. Considering the choice of flows as in the proof of Theorem~\ref{thm:SWmetaconverse} with
\begin{align*}
&\lambda_{\ssf}^{(1|2)}(s_1,s_2,x_2,y_1,y_2,\shat_1,\shat_2)=-\biggl[ \phi^{(1|2)}(s_1,s_2)\Ibb\{s_1= \shat_1\}\times \\&\Ibb\{x_2=y_2\} +\alpha \phih(s_1,s_2)\biggr] \Ibb\{(s_1,s_2)=(\shat_1,\shat_2)\},\\
&\lambda_{\ssf}^{(2|1)}(s_1,s_2,x_1,y_1,y_2,\shat_1,\shat_2)=-\biggl[ \phi^{(2|1)}(s_1,s_2)\Ibb\{s_2= \shat_2\}\times \\& \Ibb\{x_1=y_1\} +(1-\alpha) \phih(s_1,s_2)\biggr] \Ibb\{(s_1,s_2)=(\shat_1,\shat_2)\},
\end{align*} for $\alpha \in (0,1)$,
 it is easy to verify that our metaconverse in \eqref{eq:SWmetaconverse} follows via the construction of a feasible solution to the optimization problem in \eqref{eq:avenue1}. Note that here, $\sum_{s_1}\lambda_{\ssf}^{(1|2)}$ in general \textit{depends} on $\shat_1$ and $\sum_{s_2}\lambda_{\ssf}^{(2|1)}$ \textit{depends} on $\shat_2$, whereby this construction is not feasible for \eqref{eq:level3bound}.  
 On the other hand, we find that the derivation of the Miyake and Kanaya converse from \eqref{eq:MKimprovement} corresponds to the following choice of variables,
\begin{align}
&\lambda_{\ssf}^{(1|2)}(s_1,s_2,x_2,y_1,y_2,\shat_1,\shat_2)=-\Ibb\{(s_1,s_2)=(\shat_1,\shat_2)\}\non\\&\qquad \biggl[\Ibb\{y_2=x_2\} P_{S_2}(s_2)\frac{\exp(-\beta)}{M_1}\biggr],\label{eq:mkvariables}\\
&\lambda_{\ssf}^{(2|1)}(s_1,s_2,x_1,y_1,y_2,\shat_1,\shat_2)=-\Ibb\{(s_1,s_2)=(\shat_1,\shat_2)\}\non\\&\qquad \biggl[\Ibb\{y_1=x_1\} P_{S_1}(s_1)\frac{\exp(-\beta)}{M_2}+\frac{\exp(-\beta)}{M_1M_2}\biggr],\non
\end{align}
for $\beta>0$. In this case, $\sum_{s_1}\lambda_{\ssf}^{(1|2)}$ is independent of $\shat_1$ and $\sum_{s_2}\lambda_{\ssf}^{(2|1)}$ is independent of $\shat_2$. Hence, with the following choice for the remaining dual variables satsifying (D4)--(D7),
\begin{align*}
&\lambda_{\csf}(s_1,s_2,x_1,x_2,y_1,y_2)\hspace{-0.05cm}=\hspace{-0.05cm}\Ibb\{(y_1,y_2)=(x_1,x_2)\}\hspace{-0.08cm} P_{S_1,S_2}(s_1,s_2)\times\\
&\Ibb\biggl \lbrace P_{S_1,S_2}(s_1,s_2) \leq \max\biggl \lbrace P_{S_2}(s_2)\frac{\exp(-\beta)}{M_1},\\&\hspace{4cm}P_{S_1}(s_1)\frac{\exp(-\beta)}{M_2}, \frac{\exp(-\beta)}{M_1M_2}\biggr \rbrace \biggr \rbrace,
\end{align*} 
\begin{align*}
\mu_{\csf}^{(2)}(s_2,x_2,y_1,y_2)&=-\frac{\exp(-\beta)}{M_1}P_{S_2}(s_2)\Ibb\{y_2=x_2\},\\
\mu_{\csf}^{(1)}(s_1,x_1,y_1,y_2)&=-\frac{\exp(-\beta)}{M_2}P_{S_1}(s_1)\Ibb\{y_1=x_1\},\\
\mu_{\ssf}^{(1)}(s_1,\shat_1,\shat_2,y_1,y_2)&=-\frac{\exp(-\beta)}{M_1M_2}\Ibb\{s_1=\shat_1\},\\
\mu_{\csf}^{(2|1)}(s_1,s_2,x_2)&=\sum_{y_1,y_2}\lambda_{\csf}(s_1,s_2,x_1,x_2,y_1,y_2),
\end{align*}
and $ \mu_{\ssf}^{(2)}, \mu_{\csf}^{(1|2)}\equiv 0$,
 it becomes clear that the resulting Miyake-Kanaya converse follows as a lower bound on the lower level optimization problem in  \eqref{eq:level3bound}. This also implies that the converse of Miyake and Kanaya can be thought to be obtained by the $\lambda$'s in DPSW as inducing source-coding like problems  in the DP's in \eqref{eq:level3bound}. On the other hand our metaconverse in \eqref{eq:SWmetaconverse} follows from a more complicated bound.
 
 In summary, our metaconverse in \eqref{eq:SWmetaconverse} and the Miyake-Kanaya converse in \eqref{eq:MKconverse} can be placed in the hierarchy of lower bounds as illustrated in Fig~\ref{fig:hierarchy}. Moreover, this hierarchy also provides structured avenues for obtaining tighter bounds on the finite blocklength Slepian-Wolf coding problem --  by appropriately bounding optimization problems lying higher in the hierarchy in \eqref{eq:avenue11}.
\section{Conclusion}
We presented a new finite blocklength converse for the Slepian-Wolf coding problem which improves on the converse of Miyake and Kanaya. The converse was derived by employing the linear programming based framework discussed in \cite{jose2016linear}. The proposed framework was shown to imply new metaconverses for lossy source coding and lossless source coding with side information problems, and recover the tightest hypothesis testing based converse of Kostina and Verd\'{u} \cite{kostina2012fixed}. For finite blocklength Slepian-Wolf coding, a systematic approach was developed to synthesize new LP-based converses from those of lossless source coding problems with  side information. By appropriately combining the metaconverses for these point-to-point problems, our metaconverse for Slepian-Wolf coding was derived.
\section{Appendices}\label{sec:appendices}
\appendices
\section{Hypothesis Testing Based Converse for Lossy Source Coding}
For a source $S$ with distribution $P_S$, distortion function $d:\Sscr \times \Sscrhat \rightarrow [0,+\infty]$ and distortion level $\dbf$, the rate-distortion function is defined as \begin{align}R_S(\dbf)=\inf_{P_{\Shat|S}:\Ebb[d(S,\Shat)]\leq \dbf}I(S;\Shat),\label{eq:ratedistortion}
\end{align}where the infimum is over $P_{\Shat|S}\in \Pscr(\Sscrhat|\Sscr)$. Assume that the infimum in \eqref{eq:ratedistortion} is achieved by a unique $P_{\Shat^{*}|S}$ and $\dbf_{\rm min}=\inf\{\dbf:R_S(\dbf)< \infty\}$. 
The hypothesis testing based converse of Kostina and Verd\'{u} \cite[Theorem 8]{kostina2012fixed} is then obtained as below.
\begin{converse}[KV-hypothesis testing]\label{conv:KVhypoth}
Consider problem SC with $\Xscr=\Yscr=\{1,\hdots,M\}$. Any code $(f,g)$ such that $\Ebb[\Ibb\{d(S,\Shat)>\dbf\}]\leq \epsilon$ and $\dbf>\dbf_{\min}$ must satisfy,
\begin{align}
M \geq \sup_{Q \in \Pscr(\Sscr)}\inf_{\shat \in \Sscrhat} \frac{\beta_{1-\epsilon}(P_S,Q)}{M\mathbb{Q}[d(S,\shat)\leq \dbf]}, \label{eq:Mhypoth}
\end{align} where $\beta_{1-\epsilon}(P_S,Q)$ is the minimum type-II error $\sum_s Q(s) T(s)$ over all tests $T$ such that the type-I error, $\sum_s P(s)(1-T(s)) \leq \epsilon$.
Moreover, the converse in \eqref{eq:Mhypoth} is equivalent to the following lower bound on the probability of error (see \cite[Equation 72]{vazquez2016bayesian}),
\begin{align}
\epsilon \geq \sup_{Q_S \in \Pscr(\Sscr)} \alpha_{M^{*}}(P_S,Q_S),
\end{align}with $M^{*}=M \max_{\shat}\Qbb[d(S,\shat)\leq \dbf]$.
\end{converse}
\begin{corollary}\label{cor:alphahypoth}
The following relationship holds,
\begin{align}
&\sup_{Q_S \in \Pscr(\Sscr)}\biggl \lbrace \alpha_{M^{*}}(P_S,Q_S)\biggr \rbrace \non\\& =\sup_{Q_S \in \Pscr(\Sscr)}\sup_{\beta\geq 0}\biggl \lbrace \sum_s \min\{P_S(s), \beta Q_S(s)\}-\beta M^{*}\biggr \rbrace.\label{eq:alphahypothconverse}
\end{align}
\end{corollary}
\begin{proof}
To see the above equivalence, we consider the Neyman-Pearson (NP) optimal test for  $\alpha_{M^{*}}(P_S,Q_S)$. The NP optimal test is given by $T^{*}(s)=\Ibb\{\frac{P_S}{Q_S}(s)\leq \gamma^{*}\}$ such that
\begin{align}
&\alpha_{M^{*}}(P_S,Q_S)=\sum_s P_S(s)\Ibb \biggl \lbrace \frac{P_S}{Q_S}(s)\leq \gamma^{*}\biggr \rbrace \quad \mbox{and} \label{eq:alpha}\\
&\sum_s Q_S(s)\Ibb\biggl \lbrace \frac{P_S}{Q_S}(s)\leq \gamma^{*}\biggr \rbrace=1-M^{*} \label{eq:beta}.
\end{align}
Now, consider $\alpha_{M^{*}}(P_S,Q_S)-\gamma^{*}(1-M^{*})$ which evaluates to
\begin{align*}
&\sum_s P_S(s)\Ibb \biggl \lbrace \frac{P_S}{Q_S}(s)\leq \gamma^{*}\biggr \rbrace -\gamma^{*} \sum_s Q(s)\Ibb\biggl \lbrace \frac{P_S}{Q_S}(s)\leq \gamma^{*}\biggr \rbrace\\ &=\sum_s P_S(s)\Ibb \biggl \lbrace \frac{P_S}{Q_S}(s)\leq \gamma^{*}\biggr \rbrace-\gamma^{*}\\& \hspace{0.4cm}+\gamma^{*} \sum_s Q_S(s)\Ibb\biggl \lbrace \frac{P_S}{Q_S}(s)> \gamma^{*}\biggr \rbrace\\&=\sum_s \min\{P_S(s),\gamma^{*} Q(s)\}-\gamma^{*},
\end{align*}which implies that,
$$\alpha_{M^{*}}(P_S,Q_S)=\sum_s \min\{P_S(s),\gamma^{*} Q(s)\}-\gamma^{*} M^{*}.$$ Moreover, the RHS of the above equality can be equivalently written as, 
\begin{align*}&\sum_s \min\{P_S(s),\gamma^{*} Q(s)\}-\gamma^{*} M^{*}\\&=\sup_{\beta\geq 0}\left \{\sum_s \min\{P_S(s),\beta Q(s)\}-\beta M^{*}\right \}.\end{align*} The proof for the last equality follows in the same line as in the proof of \cite[Lemma 1]{elkayam15calc} and we skip the proof here. Now, taking the supremum over $Q_{S} \in \Pscr(\Sscr)$ yields the required result.
\end{proof}
\section{ Proofs of Theorems in Section~\ref{sec:sythesize}}\label{sec:appB}
\begin{proofarg}[of Proposition~\ref{prop:sideconverse}]
Let $(\bar{\gamma}^a,\bar{\gamma}^b,\bar{\lambda}^{(1|2)}_{\ssf},\bar{\lambda}^{(1|2)}_{\csf})\in \FEA(\DPSI_{1|2})$. We now show that the choice of dual variables in \eqref{eq:constr1} is feasible for DPSW.  We first verify the feasibility of the  choice of dual variables with respect to constraint (D1) of DPSW. We get that,
\begin{align*}
&\sum_{y_1,y_2}\mu^{(1)}_{\csf} (s_1,x_1,y_1,y_2)+\sum_{s_2}\mu^{(1|2)}_{\csf}(x_1,s_1,s_2)\\&=\sum_{s_2}\sum_{y_1}\bar{\lambda}^{(1|2)}_{\csf}(x_1,s_1,s_2,y_1)\stackrel{(c)}{\geq} \bar{\gamma}^a(s_1)=\gamma^a(s_1),
\end{align*}thereby satisfying (D1). The inequality in (c) follows from the constraint (B1) of $\DPSI_{1|2}$. For checking feasibility with respect to constraint (D2), we get that
\begin{align*}
&\sum_{y_1,y_2}\mu^{(2)}_{\csf}(x_2,s_2,y_1,y_2)+\sum_{s_1}\mu^{(2|1)}_{\csf}(x_2,s_1,s_2)\\&= \sum_{y_1,y_2} \bar{\gamma}^b(s_2,y_1)\Ibb\{x_2=y_2\}=\sum_{y_1}\bar{\gamma}^b(s_2,y_1)=\gamma^b(s_2),
\end{align*}thereby satisyfing (D2). The feasibility with respect to (D3) is trivially satisfied.
For feasibility with respect to (D4), the LHS of (D4) becomes
\begin{align*}
&\lambda^{(1|2)}_{\ssf}(s_1,s_2,x_2,y_1,y_2,\shat_1,\shat_2)+\lambda_{\csf} (s_1,s_2,x_1,x_2,y_1,y_2)\\&=\bar{\lambda}^{(1|2)}_{\ssf}(s_1,s_2,\shat_1,y_1)\Ibb\{x_2=y_2\}\\& \qquad +\bar{\lambda}^{(1|2)}_{\csf}(x_1,s_1,s_2,y_1)\Ibb\{x_2=y_2\}\\& \stackrel{(a)}{\leq} P(s_1,s_2)\Ibb\{y_1=x_1\}\Ibb\{y_2=x_2\}\Ibb\{s_1\neq \shat_1\}\\&\leq P(s_1,s_2)\Ibb\{(y_1,y_2)=(x_1,x_2)\}\Ibb\{(s_1,s_2)\neq (\shat_1,\shat_2)\},
\end{align*}which is the RHS, thereby satisfying (D4). Here, the inequality in (a) follows from the constriant (B3) of $\DPSI_{1|2}$.
To verify feasibility with respect to (D5), we have, $\mu^{(2)}_{\ssf}(s_2,\shat_1,\shat_2,y_1,y_2)+ \mu^{(2)}_{\csf}(x_2,s_2,y_1,y_2)=$
\begin{align*}
& \bar{\gamma}^b(s_2,y_1)\Ibb\{x_2=y_2\}  \stackrel{(b)}{\leq} \sum_{s_1}\bar{\lambda}^{(1|2)}_{\ssf}(s_1,s_2,\shat_1,y_1)\Ibb\{x_2=y_2\}\\&=\sum_{s_1}\lambda^{(1|2)}_{\ssf} (s_1,s_2,x_2,y_1,y_2,\shat_1,\shat_2),
\end{align*}which is the RHS of (D5), thereby satisfying it. The inequality in (b) follows from constraint (B2) of $\DPSI_{1|2}$. Since $\lambda^{(2|1)}_{\ssf},\mu^{(1)}_{\ssf},\mu^{(1)}_{\csf} \equiv 0$,  the constraint (D6) is trivially satisfied. To verify feasibility with respect to (D7), we have, $\mu^{(2|1)}_{\csf} (x_2,s_1,s_2)+\mu^{(1|2)}_{\csf}(x_1,s_1,s_2)=$
\begin{align*}
&\sum_{y_1}\bar{\lambda}^{(1|2)}_{\csf}(x_1,s_1,s_2,y_1)=\sum_{y_1,y_2}\bar{\lambda}^{(1)}_{\csf}(x_1,s_1,s_2,y_1)\Ibb\{x_2=y_2\}\\&=\sum_{y_1,y_2}\lambda_{\csf}(s_1,s_2,x_1,x_2,y_1,y_2),
\end{align*}thereby satisfying (D7).
 Thus, the considered choice of dual variables is feasible for DPSW.
\end{proofarg}
\begin{proofarg} [of Theorem~\ref{thm:feasDPSW}]
It is enough to show that the above choice of dual variables are feasible with respect to the constraints (D4)-(D7) of DPSW. To verify the feasibility of dual variables with respect to (D4), consider the following two cases.\\
Case 1: $\Ibb\{(s_1,s_2)\neq (\shat_1,\shat_2)\}=1$.\\
In this case, $\lambda^{(1|2)}_{\ssf}(s_1,s_2,x_2,y_1,y_2,\shat_1,\shat_2)=0$ and $\lambda^{(2|1)}_{\ssf}(s_1,s_2,x_1,y_1,y_2,\shat_1,\shat_2)=0$. The LHS of (D4) becomes,
$
\lambda_{\csf}(s_1,s_2,x_1,x_2,y_1,y_2)\leq P(s_1,s_2)\Ibb\{(y_1,y_2)=(x_1,x_2)\},
$ which is the RHS of (D4) thereby satisfying the constraint.
\\
Case 2: $\Ibb\{(s_1,s_2)\neq (\shat_1,\shat_2)\}=0$.
\\
In this case, $s_1=\shat_1$, $s_2=\shat_2$, the RHS of (D4) is zero and the LHS becomes,
\begin{align*}
&\lambda^{(1|2)}_{\ssf}(s_1,s_2,x_2,y_1,y_2,\shat_1,\shat_2)+\lambda^{(2|1)}_{\ssf}(s_1,s_2,x_1,y_1,y_2,\shat_1,\shat_1)\\&\qquad +\lambda_{\csf}(s_1,s_2,x_1,x_2,y_1,y_2)\\
&=\lambdab^{(1|2)}_{\ssf}(s_1,s_2,\shat_1,y_1)\Ibb\{x_2=y_2\}+\alpha \lambdah_{\ssf}(s_1,s_2,\shat_1,\shat_2,y_1,y_2)\\&\hspace{0.1cm}+\lambdat^{(2|1)}_{\ssf}(s_1,s_2,\shat_2,y_2)\Ibb\{x_1=y_1\}+\lambda_{\csf}(s_1,s_2,x_1,x_2,y_1,y_2) \\&\qquad+(1-\alpha)\lambdah_{\ssf}(s_1,s_2,\shat_1,\shat_2,y_1,y_2)
\end{align*}
\begin{align*}
&\leq [\lambdab^{(1|2)}_{\ssf}(s_1,s_2,\shat_1,y_1)+\lambdab^{(1|2)}_{\csf}(s_1,s_2,x_1,y_1)]\Ibb\{x_2=y_2\}\\&+[\lambdat^{(2|1)}_{\ssf}(s_1,s_2,\shat_2,y_2)+\lambdat^{(2|1)}_{\csf}(s_1,s_2,x_2,y_2)]\Ibb\{x_1=y_1\}\\& +\lambdah_{\csf}(s_1,s_2,x_1,x_2,y_1,y_2)+ \lambdah_{\ssf}(s_1,s_2,\shat_1,\shat_2,y_1,y_2)
\end{align*}which is non-positive, thereby satisfying the constraint (D4).
The non-positivity follows since $(\lambdab^{(1|2)}_{\ssf},\lambdab^{(1|2)}_{\csf})$, $(\lambdat^{(2|1)}_{\ssf},\lambdat^{(2|1)}_{\csf})$ and $(\lambdah_{\ssf},\lambdah_{\csf})$ satisfy the constraints (B3), (C3) and (A3) (corresponding to the case when $(s_1,s_2)=(\shat_1,\shat_2)$) of dual programs $\DPSI_{1|2}$, $\DPSI_{2|1}$ and DPJE respectively.

To verify feasibility with respect to (D5), $\sum_{s_1}\lambda^{(1|2)}_{\ssf}(s_1,s_2,x_2,y_1,y_2,\shat_1,\shat_2)$ evaluates to
\begin{align*}
&\bigl[\lambdab^{(1|2)}_{\ssf}(\shat_1,s_2,\shat_1,y_1)\Ibb\{x_2=y_2\} +\alpha \lambdah_{\ssf}(\shat_1,s_2,\shat_1,\shat_2,y_1,y_2)\bigr]\\& \hspace{3cm}\times\Ibb\{s_2=\shat_2\}\\&\stackrel{(a)}{\geq} [\gammab^b(s_2,y_1)-\hspace{-0.2cm}\sum_{s_1\neq \shat_1}\lambdab^{(1|2)}_{\ssf}(s_1,s_2,\shat_1,y_1)]\Ibb\{x_2=y_2,s_2=\shat_2\}\\&\qquad +\alpha\lambdah_{\ssf}(\shat_1,s_2,\shat_1,\shat_2,y_1,y_2)\Ibb\{s_2=\shat_2\}
\\&\geq \mu^{(2)}_{\csf}(x_2,s_2,y_1,y_2)+\mu^{(2)}_{\ssf}(s_2,\shat_1,\shat_2,y_1,y_2),
\end{align*}thereby satisfying (D5). The inequality in $(a)$ results from the constraint (A2).

To verify the feasibility with respect to (D6),
$\sum_{s_2}\lambda^{(2|1)}_{\ssf}(s_1,s_2,x_1,y_1,y_2,\shat_1,\shat_2)$ evaluates to
\begin{align*}
&\biggl[\lambdat^{(2|1)}_{\ssf}(s_1,\shat_2,\shat_2,y_2)\Ibb\{x_1=y_1\}\\&\hspace{1cm} +(1-\alpha)\lambdah_{\ssf}(s_1,\shat_2,\shat_1,\shat_2,y_1,y_2) \biggr] \Ibb\{s_1=\shat_1\},\\&\stackrel{(b)}{\geq} \biggl[\gammat^b(s_1,y_2)\hspace{-0.04cm}-\hspace{-0.14cm}\sum_{s_2\neq \shat_2}\hspace{-0.1cm}\lambdat^{(2|1)}_{\ssf}(s_1,s_2,\shat_2,y_2)\biggr]\Ibb\{x_1=y_1,s_1=\shat_1\}\\& \qquad + (1-\alpha)\lambdah_{\ssf}(s_1,\shat_2,\shat_1,\shat_2,y_1,y_2)\Ibb\{s_1=\shat_1\}\\
&\geq \mu^{(1)}_{\csf}(x_1,s_1,y_1,y_2)+\mu^{(1)}_{\ssf}(s_1,\shat_1,\shat_2,y_1,y_2),
\end{align*}thereby satisfying (D6). The inequality in (b) follows from the constraint (C2) of $\DPSI_{2|1}$. The feasibility with respect to (D7) is trivially satisfied.
Hence, the considered choice of dual variables are feasible for DPSW.
\end{proofarg}
\begin{proofarg}[of Theorem~\ref{thm:SWmetaconverse}]
To get to the above converse, take $(\lambdab^{(1|2)}_{\ssf},\lambdab^{(1|2)}_{\csf},\gammab^a,\gammab^b)$ as in \eqref{eq:SIDfeasible}, $(\lambdat^{(2|1)}_{\ssf},\lambdat^{(2|1)}_{\csf},\gammat^a,\gammat^b)$ as in \eqref{eq:constr4} and $(\lambdah_{\ssf},\lambdah_{\csf},\gammah^a,\gammah^b)$ as in \eqref{eq:JEmetafeasible} and substitute in \eqref{eq:DPSWfeas} to get the values of the variables $\lambda^{(1|2)}_{\ssf},\lambda^{(2|1)}_{\ssf}, \mu^{(2)}_{\ssf},$ and $\mu^{(1)}_{\ssf}$ of DPSW. For the remaining variables, choose the following values of dual variables,
\begin{align}
\lambda_{\csf}(s_1,s_2,x_1,x_2,y_1,y_2)&=\Ibb\{(y_1,y_2)=(x_1,x_2)\}\times\non\\&\hspace{-3.4cm} \min\{P_{S_1,S_2}(s_1,s_2),\phih(s_1,s_2)+\phi^{(1|2)}(s_1,s_2)+\phi^{(2|1)}(s_1,s_2)\},
\non\\
\mu^{(2)}_{\csf}(s_2,x_2,y_1,y_2)&=-\max_{\shat_1} \phi^{(1|2)}(\shat_1,s_2)\Ibb\{x_2=y_2\},\non\\
\mu^{(1)}_{\csf}(s_1,x_1,y_1,y_2)&=-\max_{\shat_2} \phi^{(2|1)}(s_1,\shat_2) \Ibb\{x_1=y_1\},\non\\
\mu^{(2|1)}_{\csf}(x_2,s_1,s_2)&\equiv \min\{P_{S_1,S_2}(s_1,s_2),\label{eq:constr6}\\&\hspace{-0.8cm}\phih(s_1,s_2)+\phi^{(1|2)}(s_1,s_2)+\phi^{(1|2)}(s_1,s_2)\},\non
\\
\gamma^c(y_1,y_2)&=-\max_{\shat_1,\shat_2}\phih(\shat_1,\shat_2),\non\\
\gamma^b(s_2)&=-M_1\max_{\shat_1}\phi^{(1|2)}(\shat_1,s_2)+\non\\& \qquad \sum_{s_1}\mu^{(2|1)}_{\csf}(x_2,s_1,s_2),\non\\
\gamma^a(s_1)&=-M_2 \max_{\shat_2}\phi^{(2|1)}(s_1,\shat_2),\non
\end{align}$\mu^{(1|2)}_{\csf}(x_1,s_1,s_2)\equiv 0$. 
The above choice of variables can be easily verified to satisfy the constraints in \eqref{eq:DPSWfeas}. 
\end{proofarg}
\bibliographystyle{IEEEtran}
\bibliography{apsbib,ref}
\end{document}